\documentclass{eptcs-modified}

\usepackage{amsmath}
\usepackage{amssymb}
\usepackage{url}
\usepackage{amsthm}
\usepackage{graphicx}
\usepackage{amscd}
\usepackage{enumitem}

\usepackage{pgf} 
\usepackage{tikz} 

\DeclareMathOperator{\rank}{rank}

\begin{document}

\theoremstyle{definition}
\newtheorem{theorem}{Theorem}[section]
\newtheorem*{theorem*}{Theorem}
\newtheorem{definition}[theorem]{Definition}
\newtheorem{problem}[theorem]{Problem}
\newtheorem{assumption}[theorem]{Assumption}
\newtheorem{corollary}[theorem]{Corollary}
\newtheorem{proposition}[theorem]{Proposition}
\newtheorem{example}[theorem]{Example}
\newtheorem{lemma}[theorem]{Lemma}
\newtheorem{observation}[theorem]{Observation}
\newtheorem{fact}[theorem]{Fact}
\newtheorem{question}[theorem]{Question}
\newtheorem{conjecture}[theorem]{Conjecture}
\newtheorem{addendum}[theorem]{Addendum}
\newtheorem{remark}[theorem]{Remark}

\newcommand{\dom}{\operatorname{dom}}
\newcommand{\fr}{\mbox{}^\smallfrown}
\newcommand{\om}{\omega}
\newcommand{\mto}{\rightrightarrows}

\newcommand{\N}{\mathbb{N}}
\newcommand{\Q}{\mathbb{Q}}
\newcommand{\R}{\mathbb{R}}
\newcommand{\PI}[2]{\ensuremath{\boldsymbol\Pi^{#1}_{#2}}}
\newcommand{\SI}[2]{\ensuremath{\boldsymbol\Sigma^{#1}_{#2}}}
\newcommand{\DE}[2]{\ensuremath{\boldsymbol\Delta^{#1}_{#2}}}

\newcommand{\Sierp}{Sierpi\'nski }
\newcommand{\leqW}{\leq_{\textrm{W}}}
\newcommand{\leqsW}{\leq_{\textrm{sW}}}
\newcommand{\nleqW}{\nleq_{\textrm{W}}}
\newcommand{\nleqsW}{\nleq_{\textrm{sW}}}
\newcommand{\leW}{<_{\textrm{W}}}
\newcommand{\equivW}{\equiv_{\textrm{W}}}
\newcommand{\lesW}{<_{\textrm{sW}}}
\newcommand{\equivsW}{\equiv_{\textrm{sW}}}
\newcommand{\equivT}{\equiv_{\textrm{T}}}
\newcommand{\geqW}{\geq_{\textrm{W}}}
\newcommand{\pipeW}{|_{\textrm{W}}}

\newcommand{\Can}{\ensuremath{{2^\N}}}
\newcommand{\Bai}{\ensuremath{{\N^\N}}}
\newcommand{\LL}{\ensuremath{\mathcal{L}_2}}
\newcommand{\Seq}{\ensuremath{\N^{<\N}}}
\newcommand{\WO}{\ensuremath{\mathbf{WO}}}
\newcommand{\LO}{\ensuremath{\mathbf{LO}}}
\newcommand{\Tr}{\ensuremath{\mathbf{Tr}}}

\newcommand{\Wif}[3]{[\textsf{if }#1\textsf{ then }#2\textsf{ else }#3]}
\newcommand{\idStB}{\left ( \id : \mathbb{S}_{\SI11} \to \mathbb{S}_\mathcal{B} \right )}
\newcommand{\id}{\textsf{id}}
\newcommand{\C}{\textsf{C}}
\newcommand{\UC}{\textsf{UC}}
\newcommand{\TC}{\textsf{TC}}
\newcommand{\TUC}{\textsf{TUC}}
\newcommand{\lpo}{\textsf{LPO}}
\newcommand{\llpo}{\textsf{LLPO}}
\newcommand{\CWO}{\ensuremath{\mathsf{CWO}}}
\newcommand{\WCWO}{\ensuremath{\mathsf{WCWO}}}
\newcommand{\set}[2]{\left\{\,{#1} \,:\, {#2}\,\right\}}
\newcommand{\SiSep}[2]{\ensuremath{\SI{#1}{#2}\text{-}\mathsf{Sep}}}
\newcommand{\PiSep}[2]{\ensuremath{\PI{#1}{#2}\text{-}\mathsf{Sep}}}
\newcommand{\DeCA}[2]{\ensuremath{\DE{#1}{#2}\text{-}\mathsf{CA}}}
\newcommand{\SiCA}[2]{\ensuremath{\SI{#1}{#2}\text{-}\mathsf{CA}}}
\newcommand{\wDeCA}[2]{\ensuremath{\DE{#1}{#2}\text{-}\mathsf{CA^-}}}
\newcommand{\wSiCA}[2]{\ensuremath{\SI{#1}{#2}\text{-}\mathsf{CA^-}}}
\newcommand{\SiUC}{\ensuremath{\SI11\text{-}\UC_\Bai}}
\newcommand{\SiWKL}{\ensuremath{\SI11\text{-}\mathsf{WKL}}}
\newcommand{\FindWSS}{\ensuremath{\mathsf{FindWS}_{\boldsymbol\Sigma}}}
\newcommand{\FindWSP}{\ensuremath{\mathsf{FindWS}_{\boldsymbol\Pi}}}
\newcommand{\FindWSD}{\ensuremath{\mathsf{FindWS}_{\boldsymbol\Delta}}}
\newcommand{\DetD}{\textsf{Det}_{\boldsymbol\Delta}}
\newcommand{\DetS}{\textsf{Det}_{\boldsymbol\Sigma}}
\newcommand{\List}{\ensuremath{\mathsf{List}}}
\newcommand{\wList}{\ensuremath{\mathsf{wList}}}
\newcommand{\PTT}[1]{\ensuremath{\mathsf{PTT}_{#1}}}
\newcommand{\wPTT}[1]{\ensuremath{\mathsf{wPTT}_{#1}}}
\newcommand{\chiP}{\ensuremath{\chi_{\Pi^1_1}}}

\title{Searching for an analogue of $\mathrm{ATR}_0$\\ in the Weihrauch lattice}

\author{
Takayuki Kihara
\institute{Department of Mathematical Informatics\\ Nagoya University, Nagoya, Japan}
\email{kihara@i.nagoya-u.ac.jp}
\and
Alberto Marcone
\institute{Dipartimento di Scienze Matematiche, Informatiche e Fisiche\\Universit\'a di Udine, Udine, Italy}
\email{alberto.marcone@uniud.it}
\and
Arno Pauly
\institute{Department of Computer Science\\Swansea University, Swansea, UK\\ \& \\ Department of Computer Science\\University of Birmingham, Birmingham, UK}
\email{Arno.M.Pauly@gmail.com}
}

\def\titlerunning{An analogue of $\mathrm{ATR}_0$ in the Weihrauch lattice?}
\def\authorrunning{T. Kihara, A. Marcone \& A. Pauly}
\maketitle

\begin{abstract}
There are close similarities between the Weihrauch lattice and the zoo of
axiom  systems in reverse mathematics. Following these similarities has
often allowed researchers to translate results from one setting to the
other. However, amongst the \emph{big five} axiom systems from reverse
mathematics, so far $\mathrm{ATR}_0$ has no identified counterpart in the
Weihrauch degrees. We explore and evaluate several candidates, and conclude
that the situation is complicated.
\end{abstract}

\section{Introduction}
Reverse mathematics \cite{SOSOA:Simpson} is a program to find the sufficient
and necessary axioms to prove theorems of mathematics (that can be formalized
in second-order arithmetic). For this, a base system ($\mathrm{RCA}_0$) is
fixed, and then equivalences between theorems and certain benchmark axioms
are proven. Sometimes, a careful reading of the original proof of the theorem
reveals which of the benchmark axioms are used, and the main challenge is to
show that the theorem indeed implies those axioms (hence the name
\emph{reverse} mathematics). A vast number of theorems turned out to be
equivalent to one of only five systems: $\mathrm{RCA}_0$, $\mathrm{WKL}_0$,
$\mathrm{ACA}_0$, $\mathrm{ATR}_0$ and $\PI11\mbox{-}\mathrm{CA}_0$. While
recently attention has shifted to theorems {\bf not} equivalent to one of the
\emph{big five}, the big five still occupy a central role in the endeavour.

Computational metamathematics in the Weihrauch lattice starts with the
observation that many theorems in analysis and other areas of mathematics
have $\Pi_2$-\emph{gestalt}, i.e.~are of the form $\forall x \in \mathbf{X}
(Q(x) \rightarrow \exists y \in \mathbf{Y} \ P(x,y))$, and can hence be seen
as computational tasks: Given some $x \in \mathbf{X}$ satisfying $Q(x)$, find
a suitable witness $y \in \mathbf{Y}$. This task can also be viewed as a
multivalued partial function $f: \subseteq \mathbf{X} \mto \mathbf{Y}$, and
thus the precise definition of Weihrauch reducibility (given in
\S\ref{Weihrauchred} below) deals with this kind of objects. Often, the task
cannot be solved algorithmically (equivalently, the multivalued function is
not computable). The research programme (as formulated by Gherardi and
Marcone \cite{gherardi}, Pauly \cite{paulyincomputabilitynashequilibria,
paulyphd} and in particular Brattka and Gherardi \cite{brattka2,brattka3}) is
to compare the degree of impossibility as follows: Assume we had a black box
to solve the task for Theorem B. Can we solve the task for Theorem A using
the black box exactly once? If so, then $A \leqW B$, $A$ is Weihrauch
reducible to $B$.

As provability in $\mathrm{RCA}_0$ is closely linked to computability, it is
maybe not that surprising that very often, classification in reverse math can
be translated easily into Weihrauch reductions\footnote{The reverse direction
would also be possible, but as reverse mathematics is the older field, occurs
seldom in practice.}. While there are a number of obstacles for precise
correspondence (see \cite{hirschfeldt} for a detailed discussion), the
resource-sensitivity of Weihrauch reductions might be the most obvious one: A
proof in reverse mathematics can use a principle multiple times, a Weihrauch
reduction uses its black box once. This obstacle does not apply to
$\mathrm{RCA}_0$ or $\mathrm{WKL}_0$ classifications.

The analogue of $\mathrm{RCA}_0$ are the computable principles, the analogue
of $\mathrm{WKL}_0$ is $\C_\Can$ (closed choice on Cantor space), and the
analogue of $\mathrm{ACA}_0$ is ${\sf lim}$ or finite iterations thereof.
Theorems equivalent to $\PI11\mbox{-}\mathrm{CA}_0$ have not yet been studied
in the Weihrauch lattice, but an obvious analogue of
$\PI11\mbox{-}\mathrm{CA}_0$ is readily defined as the function which maps a
countable sequence of trees to the characteristic function of the set of
indices corresponding to well-founded trees. This leaves $\mathrm{ATR}_0$ out
of the big five, leading Marcone to initiate the search for an analogue in
the Weihrauch lattice at a Dagstuhl meeting on Weihrauch reducibility
\cite{pauly-dagstuhl}.

Two candidates have been put forth as potential answers, $\UC_\Bai$ and
$\C_\Bai$ (unique choice and closed choice on Baire space). We will examine
some evidence for both of them, and show that the question is not as easily
answered as those for the other big five. Our main focus is on three
particular theorems equivalent to $\mathrm{ATR}_0$ in reverse mathematics:
Comparability of well orderings, open determinacy on Baire space\footnote{The
version for Cantor space has been studied in the Weihrauch degrees by Le Roux
and Pauly \cite{paulyleroux3-cie}.} and the perfect tree theorem.

\begin{theorem*}[Comparability of well orderings]
If $X$ and $Y$ are well orderings over $\N$, then $|X| \leq |Y|$ or $|Y| \leq
|X|$.
\end{theorem*}

\begin{theorem*}[Open determinacy]
Consider a two-player infinite sequential game with moves from $\N$. Let the
first player have an open winning set. Then one player has a winning
strategy.
\end{theorem*}

\begin{theorem*}[Perfect Tree Theorem]
If $T \subseteq \N^{<\N}$ is a tree, then either $[T]$ is countable or $T$
has a perfect subtree.
\end{theorem*}

\paragraph*{Structure of the paper}
In Section \ref{sec:background} we recall the prerequisite notions about
Weihrauch reducibility. While reverse mathematics serves as the motivation
for this paper, its results are not invoked in our proofs, hence we do not
expand on this area. In Section \ref{sec:choice} we recall two Weihrauch
degrees of central importance, unique choice $\UC_\Bai$ and closed choice
$\C_\Bai$ on Baire space. We then prove some equivalences to those for
variants of comprehension and separation principles. In Section
\ref{sec:Sigma-WKL}, we re-examine the strength of a separation principle,
which is shown to be equivalent to \SI11-{\sf WKL}, weak K\"onig's lemma for
$\Sigma^1_1$-trees (Theorem \ref{thm:Sigma-WKL-main}). The comparability of
well orderings is studied in Section \ref{sec:wellorders}. We see two
variants, one of which we prove to be equivalent to $\UC_\Bai$ (Theorem
\ref{thm:CWOequiv}) whereas the other resists full classification (Question
\ref{question:wcwo}).

Open determinacy and the perfect tree theorem are investigated in Sections
\ref{sec:onesided} and \ref{sec:twosided}. Both principles are formulated as
disjunctions, and the versions where we know in which case we are are proven
to be equivalent to $\UC_\Bai$ or $\C_\Bai$ in Section \ref{sec:onesided}.
The results about open determinacy can be seen as uniform versions of the
study of the complexity of winning strategies in \cite{Blass:ComWinStr}. If
no case is fixed, we arrive at Weihrauch degrees not previously studied. Some
of their properties are exhibited in Section \ref{sec:twosided}. Since the
degrees studied in Section \ref{sec:twosided} are not very well behaved, we
introduce the canonical principle $\TC_\Bai$, the total continuation of
closed choice in Section \ref{sec:tcbaire}. We prove that up to finite
parallelization, it is equivalent to the two-sided versions of open
determinacy and the perfect tree theorem, and show some additional properties
of the degree. Some concluding remarks and open questions are found in
Section \ref{sec:discussion}.

The following illustrates the strength of key benchmark principles in this article:
\[
\UC_\Bai \leW \SI11\mbox{-}{\sf WKL} \leW \C_\Bai \leW \TC_\Bai \leW \widehat{\TC_\Bai} \leW \PI11\mbox{-}{\sf CA}.
\]

\section{Background on represented spaces and Weihrauch degrees}
\label{sec:background} For background on the theory of represented spaces we
refer to \cite{pauly-synthetic}, for an introduction to and survey of
Weihrauch reducibility we point the reader to \cite{pauly-handbook}.

As usual in the area, we use angle brackets to denote a variety of pairing
and coding functions, such as those from $\N \times \N$ to $\N$, from
$\N^{<\N}$ to $\N$, and from $\N^\N \times \N^\N$, $(\N^\N)^{<\N}$ and
$(\N^\N)^\N$ to $\N^\N$. The context provides information about the one
actually employed in any given instance.

\subsection{Represented spaces}

\begin{definition}
A \emph{represented space} $\mathbf{X}$ is a set $X$ together with a partial
surjection $\delta_{\mathbf{X}} : \subseteq \Bai \to X$. If $x \in X$, any
element of $(\delta_{\mathbf{X}})^{-1} (x)$ is called a \emph{name} or a
\emph{code} for $x$.
\end{definition}

A partial function $F : \subseteq \Bai \to \Bai$ is called a \emph{realizer}
of a function $f : \subseteq \mathbf{X} \to \mathbf{Y}$ between represented
spaces, if $f(\delta_\mathbf{X}(p)) = \delta_\mathbf{Y}(F(p))$ holds for all
$p \in \dom(f \circ \delta_\mathbf{X})$. We denote $F$ being an realizer of
$f$ by $F \vdash f$. We then call $f : \subseteq \mathbf{X} \to \mathbf{Y}$
\emph{computable} (respectively \emph{continuous}), iff it has a computable
(respectively continuous) realizer.

Represented spaces can adequately model most spaces of interest in
\emph{everyday mathematics}. For our purposes, we only need a few specific
spaces that we discuss in the following, as well as some constructions of
hyperspaces.

The category of represented spaces and continuous functions is
cartesian-closed, by virtue of the UTM-theorem. Thus, for any two represented
spaces $\mathbf{X}$, $\mathbf{Y}$ we have a represented space
$\mathcal{C}(\mathbf{X},\mathbf{Y})$ of continuous functions from
$\mathbf{X}$ to $\mathbf{Y}$. The expected operations involving
$\mathcal{C}(\mathbf{X},\mathbf{Y})$ (evaluation, composition, (un)currying)
are all computable. Using the Sierpi\'nski space $\mathbb{S}$ with underlying
set $\{\top,\bot\}$ and representation $\delta_{\mathbf{S}} : \Bai \to
\{\top,\bot\}$ defined via $\delta_{\mathbf{S}}(\bot)^{-1} = \{0^\omega\}$,
we can then define the represented space $\mathcal{O}(\mathbf{X})$ of
\emph{open} subsets of $\mathbf{X}$ by identifying a subset of $\mathbf{X}$
with its (continuous) characteristic function into $\mathbb{S}$. Since
countable \emph{or} and binary \emph{and} on $\mathbb{S}$ are computable, so
are countable union and binary intersection of open sets. The space
$\mathcal{A}(\mathbf{X})$ of closed subsets is obtained by taking formal
complements, i.e.~the names for $A \in \mathcal{A}(\mathbf{X})$ are the same
as the names of $X \setminus A \in \mathcal{O}(\mathbf{X})$ (i.e.~we are
using the negative information representation).

We indicate with $\Tr$ the space of trees on $\N$ represented in an obvious
way via characteristic functions on the set of finite sequences. The
computable map $[ \ ] : \Tr \to \mathcal{A}(\Bai)$ maps a tree to its set of
infinite paths, and has a computable multivalued inverse. In other words, one
can compute a code of a tree $T$ from a code of a closed set $[T]$, and vice
versa.

Given a represented space $\mathbf{X}$ and $k\in\N$, using Borel codes, the
collections $\SI0k(\mathbf{X})$ (respectively $\PI0k(\mathbf{X})$) of $\SI0k$
(respectively $\PI0k$) subsets of $\mathbf{X}$ can be naturally viewed as a
represented space, cf.\
\cite{Bra:BorelMeas,pauly-gregoriades,pauly-ordinals}. Equivalently, we can
use the jumps of $\mathbb{S}$ to characterize these spaces. We find that
$\mathcal{A}$ and $\PI01$ (respectively $\mathcal{O}$ and $\SI01$) are
identical.

The collection $\SI11(\mathbf{X})$ of analytic subsets of $\mathbf{X}$ can
also be represented in a straightforward manner: $p$ is a name of a $\SI11$
set $S\subseteq\mathbf{X}$ iff $p$ is a name of a closed set
$P\subseteq\Bai\times\mathbf{X}$ such that $S=\{x\in\mathbf{X}:(\exists
g)\;(g,x)\in P\}$. Equivalently (\cite[Proposition
35]{pauly-descriptive-arxiv}), we can define the space $\mathbb{S}_{\SI11}$
by letting it have the underlying set $\{\top,\bot\}$, and letting $p \in
\Bai$ be a name for $\top$ iff the tree on $\N$ coded by $p$ is ill-founded;
and then identify $\SI11(\mathbf{X})$ with
$\mathcal{C}(\mathbf{X},\mathbb{S}_{\SI11})$ (here $f \in
\mathcal{C}(\mathbf{X},\mathbb{S}_{\SI11})$ represents the
$\SI11(\mathbf{X})$ set $f^{-1} (\top)$). Again, the collection
$\PI11(\mathbf{X})$ of coanalytic subsets of $\mathbf{X}$ is represented in
an obvious way by taking formal complements. We define the space
$\mathbb{S}_{\PI11}$ with underlying set $\{\top,\bot\}$, so that $p \in
\Bai$ is a name for $\top$ iff the tree on $\N$ coded by $p$ is well-founded.

We first check that basic operations on these represented spaces are
well-behaved.

\begin{lemma}
\label{lemma:sigma11truthvalues}
The following operations are computable:
\begin{enumerate}
\item $\bigvee, \bigwedge : \mathbb{S}_{\SI11}^\N \to \mathbb{S}_{\SI11}$
\item $\exists : \SI11(\mathbf{X}) \to \mathbb{S}_{\SI11}$, mapping
    non-empty sets to $\top$ and the empty set to $\bot$.
\item $\id, \neg : \mathbb{S} \to \mathbb{S}_{\SI11}$
\end{enumerate}
\begin{proof}
\begin{enumerate}
\item
For $\bigvee$, we need to show that given a sequence of trees we can compute a tree that is ill-founded iff one of the contributing trees is. This can be done by simply joining them at the root. For $\bigwedge$, we need a tree that is ill-founded iff all them are. For that, we can take the product of the trees (e.g.~as in \cite{pauly-nobrega-arxiv}).

\item From $f \in \mathcal{C}(\Bai,\mathbb{S}_{\SI11})$ we can compute by
    type-conversion some $g : \Bai \times \Bai \to \mathbb{S}$ such that
    $f(p) = \top$ iff $\exists q \in \Bai \ g(p,q) = \bot$. But then
    $\exists p \in \Bai \ f(p) = \top \Leftrightarrow \ \exists \langle p,
    q\rangle \in \Bai \ g(p,q) = \bot$, and we are done.

\item For $\neg: \mathbb{S} \to \mathbb{S}_{\SI11}$, given a name $p$ for a
    point in $\mathbb{S}$ let the tree $T$ be defined by $w \in T$ iff
    $\forall n \leq |w| \ p(n) = 0$. For $\id : \mathbb{S} \to
    \mathbb{S}_{\SI11}$, we let $T$ have only branches of the form
    $n0^\omega$, and such a branch is present iff $p(n) \neq 0$.\qedhere
\end{enumerate}
\end{proof}
\end{lemma}

\begin{proposition}\label{pro:closprop}
The following operations are computable for any represented space
$\mathbf{X}$ and $k>0$:
\begin{enumerate}
\item $\SI11(\mathbf{X})^\N \longrightarrow \SI11(\mathbf{X}),
    (A_n)_n\longmapsto\bigcup_{n\in\N}{A_n}$ (countable union);
\item $\SI11(\mathbf{X})^\N \longrightarrow \SI11(\mathbf{X}),
    (A_n)_n\longmapsto\bigcap_{n\in\N}{A_n}$ (countable intersection);
      \item $\SI11(\mathbf{X} \times \mathbf{Y}) \longrightarrow
          \SI11(\mathbf{Y}), A \longmapsto \{y \in \mathbf{Y} \mid
          \exists x \in \mathbf{X} \ (x,y) \in A\}$
\item $\SI0k(\mathbf{X}) \rightarrow \SI11(\mathbf{X})$,
    $\PI0k(\mathbf{X}) \rightarrow \SI11(\mathbf{X})$, $\SI0k(\mathbf{X})
    \rightarrow \PI11(\mathbf{X})$, $\PI0k(\mathbf{X}) \rightarrow
    \PI11(\mathbf{X})$ (inclusions);
\item $\SI0k(\Bai\times \mathbf{X})\rightarrow\SI11(\mathbf{X})$,
    $\PI0k(\Bai\times \mathbf{X})\rightarrow\SI11(\mathbf{X})$, such that
	\begin{equation*}B\mapsto A=\set{x\in X}{\exists g\in\Bai(g,x)\in B)};
	\end{equation*}
\item $\SI0k(\Bai\times \mathbf{X})\rightarrow\PI11(\mathbf{X})$,
    $\PI0k(\Bai\times \mathbf{X})\rightarrow\PI11(\mathbf{X})$, such that
	\begin{equation*}B\mapsto A=\set{x\in X}{\forall g\in\Bai(g,x)\in B};
	\end{equation*}
\item $\PI01(\Bai\times \mathbf{X})\rightarrow\PI11(\mathbf{X})$, such that
	\begin{equation*} C\mapsto A=\set{x\in X}{\exists!g\in\Bai(g,x)\in C}.
	\end{equation*}
\end{enumerate}
\begin{proof} (1-6) These all follow directly from Lemma \ref{lemma:sigma11truthvalues} together with function composition.

(7) It is well-known that $a \in \Bai$ is hyperarithmetical relative to
$\{a\} \in \PI01(\Bai)$ (cf.~Corollary \ref{corr:uc-hyppoints} and
accompanying remarks below). The section map $(x, C) \mapsto \{y \in \Bai
\mid (y,x) \in C\} : \mathbf{X} \times \PI01(\Bai \times \mathbf{X}) \to
\PI01(\Bai)$ is computable, see \cite[Proposition 4.2 (9)]{pauly-synthetic}.
Thus, we find that
$$A = \{x \in \mathbf{X} \mid \exists y \in \mathrm{HYP}(x) \ (y,x) \in C\} \cap \{x \in \mathbf{X} \mid \forall y,z ((y,x), (z,x) \in C \rightarrow \ y = z)\}.$$

The first set on the right-hand side is $\PI11$ by Kleene's
$\mathrm{HYP}$-quantification theorem \cite{kleene4,kleene5} (see also \cite[Lemma III.3.1]{sacks2}); that is, the formula $\exists y\in\mathrm{HYP}(x)\ P(x,y)$ means that there are natural numbers $a,e$ such that $a\in\mathcal{O}^x$ (which represents an ordinal $\alpha$) and the $e$-th real $\Phi_e(x^{(\alpha)})$ computable in the $\alpha$-th Turing jump of $x$ satisfies $P(x,\Phi_e(x^{(\alpha)}))$, where $\mathcal{O}^x$ is Kleene's system of ordinal notations relative to $x$ (which is a $\Pi^1_1(x)$ set), cf.~\cite{sacks2}.
This description is trivially $\Pi^1_1$, uniformly relative to $x$ and the complexity of $P$,
so that we can actually compute the $\PI11$ set from $C$. The second set
explicitly and uniformly defines a $\PI11$ set. The claim thus follows using
that intersection is a computable operation on $\PI11$ sets from (2).
\end{proof}
\end{proposition}

\begin{lemma}\label{lem:funcF}
Let $\mathbf{X}$ be a represented space. Then the function
$F:\bigsqcup_{k}\PI0k(\Bai\times \mathbf{X})\rightarrow\SI11(\mathbf{X})$
defined by
\begin{equation*}
B\mapsto A=\set{x\in \mathbf{X}}{\exists g\in\Bai (g,x)\in B},
\end{equation*}
is computable.
\begin{proof}
Proposition \ref{pro:closprop}(5) is typically proved by induction on $k$,
and the inductive argument is uniform in $k$. Since (a name for) for $B \in
\bigsqcup_{k}\PI0k(\Bai\times \mathbf{X})$ includes the information about the
$k$ such that $B \in \PI0k(\Bai\times \mathbf{X})$, we can uniformly repeat
$k$ steps of the induction argument to obtain a name for $\set{x\in
\mathbf{X}}{\exists g\in\Bai (g,x)\in B}$ as a $\SI11(\mathbf{X})$ set.
\end{proof}
\end{lemma}

We define the represented spaces $\LO$ and $\WO$ respectively of linear
orderings and countable well orderings with domain contained in $\N$ (thus
$\WO$ is a subspace of $\LO$) as follows: $p$ is a name for the linear order
$(X, {\preceq_X})$ with $X \subseteq \N$ if $p(\langle n,m \rangle) =1$ if
and only if $n \preceq_X m$. We often abuse notation by leaving $\preceq_X$
implicit and writing $X \in \LO$. We may assume without loss of generality
that, for all $X \in \LO$, $0 \notin X$ (this will be useful in Definition
\ref{def:scm} below). If $X \in \LO$ we use interchangeably $\WO(X)$ and $X
\in \WO$. If $X \in \WO$ we indicate its order type by $|X|$. Given some tree
$T \subseteq \N^{<\N}$, we define the Kleene-Brouwer ordering
$\preceq_{\rm{KB}}$ on $T$ as the transitive closure of $w \preceq_{\rm{KB}}
u$ if $w \sqsupseteq u$ and $un \preceq_{\rm{KB}} um$ if $n \leq m$. Using
the coding of finite strings we view $(T, {\preceq_{\rm{KB}}})$ as a member
of $\LO$.

\begin{observation}
\label{obs:kleenebrouwer} The map ${\sf KB} : \Tr \to \LO$ mapping a tree to
its Kleene-Brouwer ordering is computable. We have $\WO({\sf KB}(T))$ iff $T$
is well-founded.
\end{observation}

We need a technical definition, which can be found in \cite[Definition
V.6.4]{SOSOA:Simpson}, for some of our proofs related to well orderings.

\begin{definition}[double descent tree]
If $X,Y \in \LO$ the \emph{double descent tree} ${\sf T}(X,Y)$ is the set of
all finite sequences of the form $\langle (m_0,n_0), (m_1,n_1), \dots,
(m_{k-1},n_{k-1}) \rangle \in \N^{<\N}$ such that
	\begin{itemize}
		\item $m_0,m_1,\dots,m_{k-1}\in X$ and $m_0>_X m_1>_X\dots>_X m_{k-1}$,
		\item $n_0,n_1,\dots,n_{k-1}\in Y$ and $n_0>_Y n_1>_Y\dots>_Y n_{k-1}$.
	\end{itemize} We define the linear ordering $X*Y={\sf KB}({\sf T}(X,Y))$.
\end{definition}

\begin{observation}
$(X,Y) \mapsto (X*Y) : \LO \times \LO \to \LO$ is computable.
\end{observation}

With an abuse of notation, we use $\Q$ and $\N$ to denote respectively a
computable presentation of the standard linear ordering of rational numbers
and of the well ordering of natural numbers.

\begin{lemma}\label{lem:ddt} Let $X,Y \in \LO$.
	\begin{enumerate}
		\item If $\WO(X)$ then $X*Y$ and $Y*X$ are well orderings.
		\item If $\WO(X)$ and $\neg\WO(Y)$, then $|X|\leq |X*Y|$.
		\item If $\WO(Y)$, then $|X*Y|\leq|\Q*Y|$.
	\end{enumerate}
\begin{proof}
The proofs of $1$ and $2$ can be found in Lemma V.6.5 of \cite{SOSOA:Simpson}. In order to prove 3, consider a function $g:X\to\Q$  such that, for all $x,x'\in X$,
	\begin{enumerate}[label=(\alph*)]
		\item $x<_X x'\rightarrow g(x) <_\Q g(x')$,
		\item $x<_\N x'\rightarrow  g(x) <_\N g(x')$.
	\end{enumerate}
It is easy to see that such a function exists. Define then $\hat{g}:(X*Y)\to(\Q*Y)$ by putting  $\hat{g}(\langle(x_0,y_0),\dots,(x_{k-1},y_{k-1})\rangle):=\langle(g(x_0),y_0),\dots,(g(x_{k-1}),y_{k-1}))\rangle$. Property a.\ of $g$ guarantees that $\hat{g}$ is well-defined and property b.\ implies that $\hat{g}$ respects the Kleene-Brouwer orderings of the double descent trees $X*Y$ and $\Q*Y$.
\end{proof}
\end{lemma}

\subsection{Weihrauch reducibility}\label{Weihrauchred}
Intuitively, $f$ being Weihrauch reducible to $g$ means that there is an
otherwise computable procedure to solve $f$ by invoking an oracle for $g$
exactly once. We thus obtain a very fine-grained picture of the relative
strength of partial multivalued functions. Consequently, a Weihrauch
equivalence is a very strong result compared to other approaches that allow
more generous access to the principle being reduced to. 

\begin{definition}[Weihrauch reducibility]
\label{def:weihrauch} Let $f,g$ be multivalued functions on represented
spaces. Then $f$ is said to be {\em Weihrauch reducible} to $g$, in symbols
$f\leqW g$, if there are computable functions $K,H:\subseteq\Bai\to\Bai$ such
that $\left(p \mapsto K\langle p, GH(p) \rangle \right )\vdash f$ for all $G
\vdash g$.

If there are computable functions $K,H:\subseteq\Bai\to\Bai$ such that $KGH
\vdash f$ for all $G \vdash g$, then $f$ is \emph{strongly Weihrauch
reducible} to $g$, in symbols $f \leqsW g$.
\end{definition}
The relations $\leqW$, $\leqsW$ are reflexive and transitive. We use $\equivW$ ($\equivsW$) to denote equivalence
and by $\leW$ we denote strict reducibility. Both Weihrauch degrees \cite{paulyreducibilitylattice} and strong Weihrauch degrees \cite{damir} form lattices, the former being distributive and the latter not (in general, Weihrauch degrees behave more naturally than strong Weihrauch degrees).

Rather than the lattice operations, we will use two kinds of products in this
work: The parallel product $f \times g$ is just the usual cartesian product
of (multivalued) functions, which is readily seen to induce an operation on
(strong) Weihrauch degrees. We call $f$ a \emph{cylinder}, if $f \equivsW
(\id_\Bai \times f)$, and note that for cylinders, Weihrauch reducibility and
strong Weihrauch reducibility coincide.

The compositional product $f \star g$ satisfies that $$f \star g \equivW \max_{\leqW} \{f_1 \circ g_1 \mid f_1 \leqW f \wedge g_1 \leqW g\}$$
and thus is the hardest problem that can be realized using first $g$, then something computable, and finally $f$. The existence of the maximum is shown in \cite{paulybrattka4}. Both products as well as the lattice-join can be interpreted as logical \emph{and}, albeit with very different properties. The sequential product $\star$ is not commutative, however, it is the only one that admits a matching implication \cite{paulybrattka4,paulykojiro}.

Two further (unary) operations on Weihrauch degrees are relevant for us,
finite parallelization $f^*$ and parallelization $\widehat{f}$. The former
has as input a finite tuple of instances to $f$ and needs to solve all of
them, the latter takes and solves a countable sequences of instances. Both
operations are closure operators in the Weihrauch lattice. They can be used
to relax the requirement of using the oracle only once, if so desired, by
looking at the relevant quotient lattices.

In passing, we will refer to the third operation, the jump from \cite{gherardi4} (studied further in \cite{brattka15}, denoted by $f'$. We use $f^{(n)}$ to denote the result of applying the jump $n$-times. The jump only preserves strong Weihrauch degrees. The input to $f'$ is a sequence converging (with unknown speed) to an input of $f$, the output is whatever $f$ would output on the limit.

The well-studied Weihrauch degrees most relevant for us are unique closed
choice and closed choice (on Baire space), to which we dedicate the following
Section \ref{sec:choice}. Two other degrees we will refer to are $\lpo : \Bai
\to \{0,1\}$ and ${\sf lim} : \subseteq (\Bai)^\omega \to \Bai$. These are
defined via $\lpo(p) = 1$ iff $p = 0^\omega$, and ${\sf lim}((p_i)_{i \in
\N}) = \lim_{i \to \infty} p_i$. They are related by $\widehat{\lpo} \equivW
{\sf lim}$. The importance of ${\sf lim}$ is found partially in the
observation from \cite{Bra:BorelMeas} that ${\sf lim}$ is complete for Baire
class $1$ functions, and more generally, that ${\sf lim}^{(n)}$ is complete
for Baire class $n + 1$ functions.

\section{$\UC_\Bai$ and $\C_\Bai$}

The two Weihrauch degrees of central importance for this paper are unique closed choice and closed choice (on Baire space). These are defined as follows:

\label{sec:choice}
\begin{definition}\label{def:UC-and-C}
Given a represented space $\mathbf{X}$, let $\C_\mathbf{X} : \subseteq \mathcal{A}(\mathbf{X}) \mto \mathbf{X}$ be defined via $x \in \C_\mathbf{X}(A)$ iff $x \in A$ (thus, $A \in \dom(\C_\mathbf{X})$ iff $A \neq \emptyset$). Let $\UC_\mathbf{X}$ be the restriction of $\C_\mathbf{X}$ to singletons.
\end{definition}

In particular, $\UC_\mathbf{X}$ is capable of finding an element of a given $\Pi^0_1$ singleton in $\mathbf{X}$.
In \cite{pauly-ordinals} Pauly introduced the notion of iterating a Weihrauch degree $f$ over a given countable ordinal, this is denoted by $f^\dagger$. It is then shown that:
\begin{theorem}[{\cite[Theorem 80]{pauly-ordinals}}]\label{thm:Pauly-dagger}
$\UC_\Bai \equivW {\sf lim}^\dagger$
\end{theorem}

One can read the above result as a very uniform version of the famous
classical result that the Turing downward closures of $\Pi^0_1$ singletons in
$\Bai$ exhausts the hyperarithmetical hierarchy (cf.\ \cite[Corollary
II.4.3]{sacks2}).

\paragraph*{Remark:} Seeing that $\mathrm{ATR}_0$ asserts the existence of Turing jumps iterated along some countable ordinal and since ${\sf lim}$ is equivalent to the Turing jump, it may seem as if this theorem already establishes that $\UC_\Bai$ is the Weihrauch degree corresponding to $\mathrm{ATR}_0$. There is a significant difference here though in what is meant by countable ordinal: In ${\sf lim}^\dagger$, the input includes a code for something which is an ordinal in the surrounding meta-theory. In particular, any computable ordinal can be used \emph{for free}. For $\mathrm{ATR}_0$ the notion of countable ordinal is that of the model used. For example, an ill-founded computable linear order without hyperarithmetical descending chains (Kleene, see \cite[Chapter 3, Lemma 2.1]{sacks2}) counts as an ordinal in the $\om$-model ${\sf HYP}$ consisting exactly of hyperarithmetical sets, and a similar phenomenon may happen in non-$\beta$-models of $\mathrm{ATR}_0$. Things get worse if non-$\omega$-models are considered: $\mathrm{ATR}_0$ (indeed, any sound c.e.\ theory, of course) fails to prove well-foundedness of some computable ordinals.

\medskip

Note that ${\sf lim}^\dagger$ roughly corresponds to a (uniform)
hyperarithmetical reduction, and therefore Theorem \ref{thm:Pauly-dagger},
for instance, implies the following:

\begin{corollary}
\label{corr:uc-hyppoints} Whenever $\{a\} \in \mathcal{A}(\Bai)$ is
computable, then $a \in \Bai$ is hyperarithmetical.
\end{corollary}

\begin{corollary}
If $f \leqW \UC_\Bai$ for $f : \subseteq \Bai \mto \mathbf{X}$, then for
every $x \in \dom(f)$, $f(x)$ contains some $y$ hyperarithmetical relative to
$x$.
\end{corollary}

Corollary \ref{corr:uc-hyppoints} is a well-known classical fact saying that
every $\Pi^0_1$ singleton is hyperarithmetical. Indeed, Spector showed that
every $\Sigma^1_1$ singleton is hyperarithmetical (cf.\ \cite[Theorem
I.1.6]{sacks2}). Thus, it is natural to ask whether choice from $\Sigma^1_1$
singletons has exactly the same strength as $\UC_\Bai$.

One can generalize Definition \ref{def:UC-and-C} to any
$\mathbf{\Gamma}\in\{\SI{i}{k},\PI{i}{k},\DE{i}{k}\}$ in a straightforward
manner: Let $\mathbf{\Gamma}\mbox{-}\C_\mathbf{X} : \subseteq
\mathbf{\Gamma}(\mathbf{X}) \mto \mathbf{X}$ be defined via $x \in
\mathbf{\Gamma}\mbox{-}\C_\mathbf{X}(A)$ iff $x \in A$. In other words, any
realizer of $\mathbf{\Gamma}\mbox{-}\C_\mathbf{X}$ sends a code of a
$\mathbf{\Gamma}$-definition of $A$ to a name of an element of $A$. Let
$\mathbf{\Gamma}\mbox{-}\UC_\mathbf{X}$ be the restriction of
$\mathbf{\Gamma}\mbox{-}\C_\mathbf{X}$ to singletons. For instance, a
realizer for $\SI11$-unique choice $\SiUC:\subseteq\SI11(\Bai)\to\Bai$ is a
partial function which, given a $\SI11$-code of a singleton $\{x\} \subseteq
\Bai$, returns a name of its unique element $x$. We will see below (in
Theorem \ref{theo:bigucbaire}) that $\SiUC \equivW \UC_\Bai$.

We now explore the strength of $\C_\Bai$.

\begin{theorem}[Kleene \cite{kleene4}]
There exists computable non-empty $A \in \mathcal{A}(\Bai)$ containing no
hyperarithmetical point.
\end{theorem}

That is, there is a nonempty $\Pi^0_1$ set $A\subseteq\Bai$ with no
hyperarithmetical element. This shows that $\C_\Bai$ has a computable
instance with no hyperarithmetical solution. Let $\mathsf{NHA} : \Bai \mto
\Bai$ be defined via $q \in \mathsf{NHA}(p)$ iff $q$ is not hyperarithmetical
relative to $p$.

\begin{corollary}
\label{corr:nhabasic} $\mathsf{NHA} \nleqW \UC_\Bai$ but $\mathsf{NHA} \leqW
\C_\Bai$.
\end{corollary}

We now get the separation between $\UC_\Bai$ and $\C_\Bai$.

\begin{corollary}
$\UC_\Bai \leW \C_\Bai$.
\end{corollary}

There are a number of variants of unique choice, comprehension and separation
that are all equivalent to $\UC_\Bai$ w.r.t.~Weihrauch reducibility. We
explore some of these next:

\begin{definition}[$\SI11$-Separation]
Let $\SiSep11:\subseteq (\Tr\times\Tr)^{\N}\mto\Can$ be the multivalued
function with $\dom(\SiSep11)=\set{(S_n,T_n)_{n\in\N}}{\forall
n([S_n]=\emptyset\vee [T_n]=\emptyset)}$ that maps any sequence
$(S_n,T_n)_{n\in\N}$ in the domain to the set
\[
\set{f\in\Can}{\forall n \left ( ([S_n]\neq\emptyset\rightarrow f(n)=0)\wedge ([T_n]\neq\emptyset\rightarrow f(n)=1)\right )}.
\]
\end{definition}

One can introduce a similar multivalued function by directly using the space
$\SI11(\N)\times\SI11(\N)$ instead of $(\Tr\times\Tr)^\N$ without affecting
the Weihrauch degree.

\begin{definition}[$\DE11$-Comprehension]
Let $\DeCA11:\subseteq (\Tr\times\Tr)^{\N}\to\Can$ be the restriction of
$\SiSep11$ to the set $\set{(S_n,T_n)_{n\in\N}}{\forall n
([S_n]=\emptyset\leftrightarrow [T_n]\neq\emptyset)}$. Let $\wDeCA11$ be the
restriction of $\DeCA11$ to the set $\{(S_n,T_n)_{n\in\N}:\forall n\
|[S_n]|+|[T_n]|=1\}$.
\end{definition}

\begin{definition}[Weak $\SI11$-Comprehension]
Let $\wSiCA11:\subseteq \Tr^\N\to\Can$ be the function with domain
$\dom(\wSiCA11)=\set{(T_n)_{n\in\N}}{\forall n |[T_n]|\leq 1}$ and that maps
$(T_n)_{n\in\N}$ to the unique $f\in\Can$ such that $f(n)=1\leftrightarrow
|[T_n]|=1$ for all $n\in\N$.
\end{definition}

\begin{theorem}
\label{theo:bigucbaire}
The following are strongly Weihrauch equivalent:
\begin{enumerate}
\item $\UC_\Bai$
\item $\SiUC$
\item $\SiSep11$
\item $\DeCA11$
\item $\wDeCA11$
\item $\wSiCA11$
\end{enumerate}
\begin{proof}
\begin{description}
\item[\textmd{($\SiUC \leqsW \UC_\Bai$)}] The proof of  \cite[Theorem
    80]{pauly-ordinals} implicitly contains a proof of\linebreak
    $\SI11\text{-}\UC_\mathbb{N} \leqsW \lim^\dagger$ (in the last
    paragraph). It is clear that $\SiUC \equivsW
    \widehat{\SI11\text{-}\UC_\mathbb{N}}$ and that $\widehat{\UC_\Bai}
    \equivsW \UC_\Bai$, so the claim follows with Theorem
    \ref{thm:Pauly-dagger}.

     An alternative proof can be obtained by noting that the proof of
    $\UC_\Bai\leqsW\wDeCA11$ given below is readily adapted to show that
    $\SiUC \leqsW \DeCA11$ instead, and use the reductions below.

\item[\textmd{($\UC_\Bai \leqsW \SiUC$)}] Trivial, as
    $\id : \PI01(\Bai) \to \SI11(\Bai)$ is computable by Proposition
    \ref{pro:closprop}(4).

\item[\textmd{($\SiSep11 \leqsW \UC_\Bai$)}] By \cite[Proposition 62 \&
    Lemma 79]{pauly-ordinals}. An alternative proof can be obtained by
    combining Lemmata \ref{lemma:cwoupper} and \ref{lemma:CWOlower} below.

\item[\textmd{($\DeCA11 \leqsW \SiSep11$)}] The former is a restriction of the latter.

\item[\textmd{($\wDeCA11 \leqsW \DeCA11$)}] The former is a restriction of the latter.

\item[\textmd{($\UC_\Bai\leqsW\wDeCA11$)}] Let $\{f\}$ be a singleton of
    $\Bai$ given via some tree $T$ such that $[T] =\{f\}$. From $T$ we
    compute the double-sequence of trees $(T^0_t,T^1_t)_{t\in\N^{<\N}}$
    such that: for all $t\in\N^{<\N}$,
	\begin{itemize}
		\item $T^0_t= \set{s\in T}{t\sqsubseteq s\vee s\sqsubseteq t}$,
		\item $T^1_t=\set{s\in T}{t\not\sqsubseteq s}$.
	\end{itemize}
Note that, for all $t\in\N^{<\N}$, exactly one between $T^0_t$ and $T^1_t$ is ill-founded. In fact, if $t\sqsubseteq f$ then $f\in[T^0_t]$ and, since $T$ has only one path, $T^1_t$ is well-founded. Otherwise, if $t\not\sqsubseteq f$ then $f\in[T^1_t]$ and $[T^0_t]=\emptyset$. Hence, we even have that for all $t\in\Seq$, $|[T^0_t]|+|[T^1_t]|= 1$.

Since we can identify $\N^{<\N}$ with $\N$ we can consider
$g=\wDeCA11((T^0_t,T^1_t)_{t\in\N^{<\N}})$. For all $t\in\N^{<\N}$, $g(t)=0
\iff [T^0_t]\neq\emptyset \iff t\sqsubseteq f$. Therefore, given $n\in\N$,
to compute $f(n)$ it suffices to wait for the first $t\in\N^{n+1}$ such
that $g(t)=0$ and then put $f(n)=t(n)$. This concludes the proof.

\item[\textmd{($\wDeCA11\leqsW\wSiCA11$)}] For every $(T^0_n,T^1_n)_{n\in\N}\in\dom(\wDeCA11)$ we have that $\wDeCA11((T^0_n,T^1_n)_{n\in\N}) = \wSiCA11((T^1_n)_{n \in \N})$.

\item[\textmd{($\wSiCA11 \leqsW \SiUC$)}] Let
    $(T_n)_{n\in\N}$ be a sequence of trees in $\dom(\wSiCA11)$. We claim
    that using $\SiUC$ we are able to compute
    $f\in\Can$ such that:
	\begin{equation}\label{eq:wSiCACh}
		\forall n(f(n)=1\leftrightarrow |[T_n]|=1).
	\end{equation}
In fact, (\ref{eq:wSiCACh}) is equivalent to
	\begin{equation*}\forall n[(f(n)=0 \vee\exists g(g\in[T_n]))\wedge(\neg\exists!g(g\in[T_n])\vee f(n)=1)],
	\end{equation*}
which in turn is equivalent to
\begin{equation}\label{eq:wSiCACh1} \forall n[\exists g(f(n)=0\vee g\in[T_n])\wedge\neg\exists!g(g\in[T_n]\wedge f(n)=0)].
\end{equation}
Now, for each $n$, we can uniformly compute from $(T_n)_{n\in\N}$ a name
for
\[
\set{(g,f)\in\Bai\times\Bai}{f(n)=0\vee g\in [T_n]}
\]
as a closed subset of $\Bai\times\Bai$, which entails that we can uniformly
compute from $(T_n)_{n\in\N}$ a name for
\[
\set{f\in\Bai}{\exists g(f(n)=0\vee g\in [T_n])}
\]
as a $\SI11(\Bai)$ set for each $n\in\N$. Furthermore, for each $n\in\N$,
we can uniformly compute from $(T_n)_{n\in\N}$ a name for
\[
\set{(g,f)\in\Bai\times\Bai}{g\in[T_n]\wedge f(n)=0}
\]
as a closed set and hence a name for
\[
\set{f\in\Bai}{\neg\exists!g(g\in[T_n]\wedge f(n)=0)}
\]
as a $\SI11(\Bai)$ set by Proposition \ref{pro:closprop}(7).

Finally, since the operations of finite and countable intersection of
$\SI11$ sets are computable, we are able to uniformly compute from
$(T_n)_{n\in\N}$ a name (by Proposition \ref{pro:closprop}(2)) for the
$\SI11(\Bai)$ singleton
\begin{equation*}
\set{f\in\Can}{\forall n[\exists g(f(n)=0\vee g\in[T_n])\wedge\neg\exists!g(g\in[T_n]\wedge f(n)=0)]}.
\end{equation*}
Clearly, applying $\SiUC$ to such set we obtain the unique $f$ satisfying
(\ref{eq:wSiCACh}), which is exactly $\wSiCA11((T_n)_n)$.\qedhere
\end{description}
\end{proof}
\end{theorem}

\subsection*{Arithmetical transfinite recursion}

As mentioned above, the operation ${\sf lim}^\dagger$ from
\cite{pauly-ordinals} is the ordinal-iteration of the map ${\sf lim}$. Here,
we will explore a direct encoding of arithmetical transfinite recursion as a
Weihrauch degree, and give another proof of its equivalence with $\UC_\Bai$.
Let us fix  an effective enumeration $\langle \phi_n:n\in\N\rangle$ of all
the computable functions $\phi:\subseteq \Bai \to\Bai$. Note that
$\widehat{\lpo^{(k)}}$ is a complete $\Sigma^0_{k+2}$-computable function,
and thus one can think of $\theta^k_n=\widehat{\lpo^{(k)}}\circ\phi_n$ as the
$n^{th}$ $\Sigma^0_{k+2}$-computable function. Instead, we could have used
the $n^{th}$ $\Sigma^0_{k+2}$ formula to define an equivalent notion.

\begin{definition}[Arithmetical transfinite recursion]\label{def:ATR} Let $\mathsf{ATR}:\subseteq\Can\times\WO\times\N^2\to\Can$  be the function which maps each $(Z,X,(k,n))\in\Can\times\WO\times\N^2$ to the set $Y\in\Can$ such that, for all $( y,j)\in\N^2$,
	\begin{equation*}
		(y,j) \in Y\leftrightarrow j\in X\wedge y\in \theta^k_n(Y^j\oplus Z),
	\end{equation*}
where $Y^j=\set{\langle y,i\rangle\in Y}{ i<_X j}$.
\end{definition}

Compare Definition \ref{def:ATR} with ${\rm ATR}_0$ in reverse mathematics,
cf.\ \cite[Definition V.2.4]{SOSOA:Simpson}. Note that our $\mathsf{ATR}$ is
a {\em single-valued} function since, as mentioned in the first remark in
this section, our $X$ is truly well ordered, and therefore, we do not need to
consider pseudo-hierarchies.

\begin{theorem}
$\mathsf{ATR} \equivsW \UC_\Bai$.
\begin{proof}
By Lemmata \ref{lemma:atrupper}, \ref{lemma:atrlower} below and Theorem \ref{theo:bigucbaire}.
\end{proof}
\end{theorem}

The following is an analog of the classical reverse mathematical fact
\cite[Theorem V.5.1]{SOSOA:Simpson}.

\begin{lemma} $\mathsf{ATR}\leqsW\SiSep11$.
\label{lemma:atrupper}
\begin{proof}
It is easy to see that $\SiSep11$ is a cylinder and hence it suffices to show
$\mathsf{ATR}\leqW\SiSep11$. Given $(Z,X,\langle
k,n\rangle)\in\Can\times\WO\times\N^2$, we want to compute
$\mathsf{ATR}(Z,X,\langle k,n\rangle)$ as defined in Definition
\ref{def:ATR}. For each $j\in X$ and $Y\in\Can$, let us consider the
following formula:
	\begin{equation*} H(Y,j)\equiv\forall \langle y,i\rangle\in\N^2[\langle y,i\rangle\in Y\iff i<_X j\wedge y\in \theta^k_n(Y^i\oplus Z)],
	\end{equation*}
Essentially, $H(Y,j)$ says that $Y$ is the set $\set{\langle y,i\rangle\in\mathsf{ATR}(Z,X,\langle k,n\rangle)}{i<_X j}$.  Using now $H$, we define the following two formulas for each $j,z\in\N$:
	\begin{gather*}
		\varphi_0(j,z)\equiv j\in X\wedge \exists Y\in\Can[H(Y,j)\wedge z\in \theta_n^k(Y^j\oplus Z)],\\
		\varphi_1(j,z)\equiv j\in X\wedge\exists Y\in\Can[H(Y,j)\wedge z\notin \theta_n^k(Y^j\oplus Z)].
	\end{gather*}
Note that, for each $j\in X$ and $z\in\N$ we have $\varphi_0(j,z)\iff \langle z,j\rangle\in\mathsf{ATR}(Z,X,\langle k,n\rangle)$.

Using the function $F$ defined in Lemma \ref{lem:funcF} and the closure
properties of Proposition \ref{pro:closprop}, we are able to compute two
names for the $\SI11(\N^2)$-sets $A_0$ and $A_1$ corresponding to the
formulas $\varphi_0$ and $\varphi_1$. Note that in this case the use of $F$
is required and we cannot appeal to Proposition \ref{pro:closprop}(5) because
$k$ is not fixed but is given with the input. It is easy to see that $A_0$
and $A_1$ are disjoint; hence one can ask  $\SI11\mbox{-}{\sf Sep}$ to give
us $f$ separating $A_0$ from $A_1$, which is clearly a solution of
$\mathsf{ATR}(Z,X,\langle k,n\rangle)$. Here are the details:

Since the names for $A_0$ and $A_1$ are $\PI01(\Bai\times\N^2)$-names, it is not difficult to see that we can build a double sequence of trees $(T^0_{\langle j,z\rangle},T^1_{\langle j,z\rangle})_{j,z\in\N}$ such that, for each $j\in \N$ and $z\in\N$,
	\begin{itemize}
		\item $\langle j,z\rangle\in A_0\iff [T^0_{\langle j,z\rangle}]\neq\emptyset$,
		\item  $\langle j,z\rangle\in A_1\iff [T^1_{\langle j,z\rangle}]\neq\emptyset$.
	\end{itemize} Note that, if $j\notin X$ then for each $z\in\N$,
$\neg\varphi_0(j,z)$ and $\neg\varphi_1(j,z)$, which means that
$[T^0_{\langle j,z\rangle}]=[T^1_{\langle j,z\rangle}]=\emptyset$. If
instead $j\in X$ we have, for each $z\in\N$,
$\varphi_0(j,z)\iff\neg\varphi_1(j,z)$ which implies $[T^0_{\langle
j,z\rangle}]\neq\emptyset\iff [T^1_{\langle j,z\rangle}]=\emptyset$.
Therefore the double-sequence of trees $(T^0_{\langle
j,z\rangle},T^1_{\langle j,z\rangle})_{j,z\in\N}$ belongs to the domain of
$\SiSep11$. So let $f\in\SiSep11(T^0_{\langle j,z\rangle},T^1_{\langle
j,z\rangle})_{j,n\in\N}$. Now we have, for each $j\in X$ and $z\in\N$,
$f(j,z)=0\iff [T^0_{\langle j,z\rangle}]\neq\emptyset\iff\varphi_0(j,z)\iff
\langle z,j\rangle\in\mathsf{ATR}(Z,X,\langle k,n\rangle)$, i.e.~we are
able to compute $\mathsf{ATR}(Z,X,\langle k,n\rangle)\in\Can$ using $f$.

Note that we are using the original input to test whether $j \in X$.
\end{proof}
\end{lemma}

\begin{lemma} $\DeCA11\leqsW\mathsf{ATR}$.
\label{lemma:atrlower}
\begin{proof}
Let $(T^0_n,T^1_n)_{n\in\N}\in\dom(\DeCA11)$, we want to compute $f\in\Can$
such that, for all $n\in\N$, $f(n)=0\iff [T^0_n]\neq\emptyset$. In order to
apply $\mathsf{ATR}$ we have to specify a set parameter $Z$, a well ordering
$X$ and an arithmetical formula. The role of $Z$ in this case will be played
by $(T^0_n,T^1_n)_{n\in\N}$. The well ordering $X$ is obtained as $\sum_{n
\in \N} ({\sf KB}(T^0_n) * {\sf KB}(T^1_n)) + 1$ (which is a well ordering by
Lemma \ref{lem:ddt}(1)).

It remains to specify an arithmetical formula $\varphi(y,Y^j\oplus Z)$ which describes what to do at each step of the recursion. We read both $Y^j$ and $Z$ as coding a sequence of pairs of trees. The idea is to eliminate at each step the leaves of all the trees in the sequence. Thus, $\varphi(y,Y^j\oplus Z)$ holds if either $Y^j = \emptyset$ and $y$ codes a vertex with a child in $Z$, or $y$ codes a vertex with a child in each tree from $Y^j$. This is easily verified to be an arithmetical formula, and hence can be coded as some $\theta^k_n$.\footnote{Similar ideas are found in the investigation of the Weihrauch degree of the \emph{pruning derivative} of a tree in \cite{pauly-nobrega-arxiv}.}

Finally, consider $Y=\mathsf{ATR}((T^0_n,T^1_n)_n,X,\langle k,n\rangle)$,
which is the set we obtain after repeating, along the well ordering $X$, the
procedure of eliminating leaves from the trees $T^0_n$ and $T^1_n$.  Now, let
fix $n$ and consider $i\in\{0,1\}$ such that $T^i_n$ is well founded. Note
that, in order to eliminate all the tree $T^i_n$, the recursion should be
done at least over the ordinal $\rank(T^i_n)$. In our case, the recursion is
done over $X$ whose order type is greater than the order type of ${\sf
KB}(T^i_n)$ which in turn is greater than $\rank(T^i_n)$, cf.~Lemma
\ref{lem:ddt}(2). This means that $Y$ does not contain any element of the
tree $T^i_n$. This argument applies to each well founded tree in the sequence
$(T^0_n,T^1_n)_n$, so we can know whether a tree in the sequence has a path
or not simply by checking if its root is in $Y$. It is easy to see that this
allows us to compute $\DeCA11((T^0_n,T^1_n)_{n\in\N})$.
\end{proof}
\end{lemma}


\section{$\Sigma^1_1$-weak K\"onig's lemma}\label{sec:Sigma-WKL}

\subsection{$\Sigma^1_1$ versus $\Pi^1_1$}

In this section, we focus on the following contrast between reverse
mathematics and the Weihrauch lattice regarding $\Sigma^1_1$ and
$\Pi^1_1$-separation: On the one hand, in reverse mathematics, we have
\begin{align}\label{equ:sep-rev-math}
\PI11\mbox{-}{\rm SEP}_0<\SI11\mbox{-}{\rm SEP}_0
\end{align}
where ${\rm A}<{\rm B}$ indicates ${\rm RCA}_0\vdash {\rm B} \to {\rm A}$,
but ${\rm RCA}_0 \nvdash {\rm A} \to {\rm B}$. On the other hand, in the
Weihrauch lattice, we have
\begin{align}\label{equ:sep-weihrauch}
\SiSep11 \leW \PiSep11.
\end{align}
The former inequality (\ref{equ:sep-rev-math}) was proven by Montalb\'an
\cite{Mon08} using Steel's tagged tree forcing. The latter inequality
(\ref{equ:sep-weihrauch}) follows from the well-known fact in descriptive set
theory that $\SI11$ has the $\DE11$-separation property, while $\PI11$ does
not (see also Lemma \ref{lem:UC-p11Sep}). It is not hard to explain the cause
of the contrast between (\ref{equ:sep-rev-math}) and
(\ref{equ:sep-weihrauch}), namely the Spector-Gandy phenomenon.

Let $\mathcal{M}$ be an $\om$-model, and let $(\SI11)^\mathcal{M}$ be the
collection of all subsets of $\om$ which are $\SI11$-definable within
$\mathcal{M}$, that is,
$(\SI11)^\mathcal{M}=\{\{n\in\om:\mathcal{M}\models\varphi(n)\}:\varphi\in\SI11\}$.
We define $(\PI11)^\mathcal{M}$ analogously. Consider the $\om$-model ${\sf
HYP}$ consisting of all hyperarithmetical reals. The Spector-Gandy theorem
(cf.\ \cite[Theorem III.3.5 + Lemma III.3.1]{sacks2} or \cite[Theorems
VIII.3.20 + VIII.3.27]{SOSOA:Simpson}) implies that
\[
(\SI11)^{\sf HYP}=\PI11, \mbox{ and } (\PI11)^{\sf HYP}=\SI11.
\]
The roles of $\SI11$ and $\PI11$ are interchanged! We should always be
careful about this role-exchange phenomenon of $\SI11$ and $\PI11$ when
comparing reverse math and computability theory. Of course, the notion of a
$\beta$-model solves this role-exchange problem. To be precise, a {\em
$\beta$-model} (see \cite[Section VII]{SOSOA:Simpson}) is an $\om$-model
$\mathcal{M}$ satisfying the following condition:
\[
(\SI11)^\mathcal{M}=\SI11, \mbox{ and } (\PI11)^\mathcal{M}=\PI11.
\]

However, the notion of a $\beta$-model is obviously related to closed choice
$\C_\Bai$: An $\om$-model $\mathcal{M}$ is a $\beta$-model iff, for any
$Z\in\mathcal{M}$ and non-empty $\Pi^0_1(Z)$ set $P\subseteq\Bai$, some
$\alpha\in P$ belongs to $\mathcal{M}$. Therefore, when studying principles
weaker than $\C_\Bai$, we cannot work within the $\beta$-models.

Now, how should we interpret the reverse-mathematical $\SI11$-separation
principle in our real universe? {\em The} right answer may not exist. It may
be $\PiSep11$ or may be $\SiSep11$.

We have already examined the strength of the $\SI11$-separation principle
$\SiSep11$. In this section, we will investigate the $\PI11$-separation
principle, $\PiSep11$, in the Weihrauch lattice. In reverse mathematics,
Montalb\'an \cite{Mon08} showed that the strength of the $\Pi^1_1$-separation
principle is strictly between $\DE11\mbox{-}{\rm CA}_0$ and ${\rm ATR}_0$
(\footnote{Actually, Montalb\'an showed that $\PI11$-separation is strictly
weaker than $\SI11\mbox{-}{\sf AC}$.}):
\[
\DE11\mbox{-}{\rm CA}_0 < \PI11\mbox{-}{\rm SEP}_0 < {\rm ATR}_0 \equiv \SI11\mbox{-}{\rm SEP}_0.
\]

Moreover, $\DE11\mbox{-}{\rm CA}_0$ and $\PI11\mbox{-}{\rm SEP}_0$ are {\em
theories of hyperarithmetic analysis}, that is, for every $Z \subseteq \om$,
${\sf HYP}(Z)$ is the least $\omega$-model of that theory containing $Z$. On
the other hand, ${\sf HYP} \not\models {\rm ATR}_0$. In contrast, we will see
the following:
\[
\UC_\Bai \equivW \DeCA11 \equivW {\sf ATR} \equivW \SiSep11 \leW \PiSep11 \leW \C_\Bai.
\]

\subsection{The strength of $\SI11$-weak K\"onig's lemma}\label{ssec:SIWKL}
The principle of $\PI01$-separation was studied already in the precursor
works by Weihrauch \cite{Weihrauchb}, and Weak K\"onig's Lemma (aka closed
choice on Cantor space) was a focus in the earliest work on Weihrauch
reducibility in the modern understanding
\cite{gherardi,brattka2,paulybrattka}. Here, we explore their higher-level
analogues.

Let $\PiSep11$ be the following partial multivalued function: Given
$\PI11$-codes of sets $A,B\subseteq\N$, if $A$ and $B$ are disjoint, then
return a set $C\subseteq\N$ separating $A$ from $B$, that is, $A\subseteq C$
and $B\cap C=\emptyset$. To be more precise:

\begin{definition}
Let $\PiSep11 :\subseteq \PI11(\N)\times \PI11(\N) \rightrightarrows 2^\N$ be
such that $C\in \PiSep11 (A,B)$ iff $C$ separates $A$ from $B$, where
$(A,B)\in{\rm dom}(\PiSep11)$ iff $A\cap B=\emptyset$.
\end{definition}

We also consider $\SI11$-weak K\"onig's lemma $\SiWKL$: Given a $\SI11$-code
of a set $T\subseteq 2^{<\om}$, if $T$ is an infinite binary tree, then
return a path through $T$. Formally speaking:

\begin{definition}
Let $\SiWKL:\subseteq\SI11(2^{<\om})\rightrightarrows 2^\N$ be such that
$p\in\SiWKL(T)$ iff $p$ is an infinite path through $T$, where $T\in{\rm
dom}(\SiWKL)$ iff $T$ is an infinite binary tree.
\end{definition}

While $\SiWKL$ appears as a $\SI11$-version of closed choice on Cantor space,
it is not equivalent to $\SI11$-choice on $2^\N$ (nor, equivalently, closed
choice on $\Bai$). Instead, it is equivalent to the parallelization
$\widehat{\SI11\mbox{-}\C_{\mathbf{2}}}$ of $\SI11$ choice on the discrete
space ${\mathbf{2}}=\{0,1\}$. We will show the following.

\begin{theorem}\label{thm:Sigma-WKL-main}
$\UC_\Bai \leW \widehat{\SI11\mbox{-}\C_{\mathbf{2}}} \equivW \PiSep11
\equivW \SiWKL \leW \widehat{\SI11\mbox{-}\C_\N} \leqW \C_\Bai$.
\end{theorem}

We will use the following fundamental notion in HYP-theory.
A {\em $\Pi^1_1$-norm} on a $\Pi^1_1$ set $P\subseteq\N$ is a map $\varphi:\N\to\om_1^{CK}\cup\{\infty\}$ such that $P=\{n:\varphi(n)<\infty\}$ and that the following relations $\leq_\varphi$ and $<_\varphi$ are $\Pi^1_1$:
\begin{align*}
a\leq_\varphi b\iff \varphi(a)<\infty\mbox{ and }\varphi(a)\leq\varphi(b),\\
a<_\varphi b\iff\varphi(a)<\infty\mbox{ and }\varphi(a)<\varphi(b).
\end{align*}
It is well-known that every $\Pi^1_1$ set admits a $\Pi^1_1$-norm (in an
effective manner): Consider a many-one reduction from a $\Pi^1_1$ set $P$ to
the set ${\rm WO}$ of well orderings. We will explore the uniform complexity
of this kind of stage comparison principle in Section \ref{sec:wellorders}.

One can easily separate unique choice on $\Bai$ and the $\PI11$-separation
principle by considering the {\em diagonally non-hyperarithmetical}
functions, which is a HYP version of ${\sf DNC}_2$ (known as diagonally
noncomputable functions). A very basic fact in HYP-theory is the existence of
a computable enumeration $(\psi_e)_{e\in\N}$ of all partial $\Pi^1_1$
functions on $\N$. For instance, let $\psi_e$ be a standard
$\Pi^1_1$-uniformization of the $e^{th}$ $\Pi^1_1$ set
$P_e\subseteq\N\times\N$, that is, $\psi_e(n)$ is an element in the $n^{th}$
section of $P_e$ attaining the smallest $\varphi$-value if it exists, where
$\varphi$ is a $\Pi^1_1$-norm on $P_e$.

\begin{lemma}\label{lem:UC-p11Sep}
$\UC_\Bai\leW\PiSep11$.
\end{lemma}

\begin{proof}
To see that $\UC_\Bai\leqW\PiSep11$, note that $\UC_\Bai \equivW \DeCA11$ by
Theorem \ref{theo:bigucbaire}, and $\DeCA11 \leqW \PiSep11$ is
straightforward. For the separation, let $(\psi_e)_{e\in\N}$ be an
enumeration of all partial $\Pi^1_1$ functions on $\N$ as above. For $i<2$,
consider $P_i=\{e\in\N:\psi_e(e)\downarrow=i\}$. Clearly $P_i$ is $\Pi^1_1$,
and $P_0\cap P_1=\emptyset$. It is easy to see that there is no $\Delta^1_1$
set separating $P_0$ and $P_1$.
\end{proof}

The proof of Lemma \ref{lem:UC-p11Sep} motivates us to introduce the
following multivalued function $\Pi^1_1\mbox{-}{\sf DNC}_2: 2^\N
\rightrightarrows 2^\N$: Given an oracle $X$, return a two-valued
$X$-diagonally non-hyperarithmetical function $f$, that is, $f \in
\Pi^1_1\mbox{-}{\sf DNC}_2(X)$ iff, whenever $\psi_e^X(e) \downarrow$, $f(e)
\neq \psi_e^X(e)$, where $(\psi_e^X)_{e\in\N}$ is a canonical enumeration of
all partial $\Pi^1_1(X)$ functions on $\N$. The following is an analog of the
well-known fact that every ${\sf DNC}_2$-function has a {\sf PA}-degree.

\begin{proposition}\label{prop:sep-dnc}
$\PiSep11 \equivW \Pi^1_1\mbox{-}{\sf DNC}_2$.
\end{proposition}
\begin{proof}
Let $P_0$ and $P_1$ be disjoint $\Pi^1_1$ sets. Clearly there is $e$ such
that $n\in P_i$ iff $\psi_e(n)\downarrow=i$. By the recursion theorem, one
can uniformly find a computable function $r$ such that $\psi_{r(n)}(r(n))
\simeq \psi_e(n)$. Let $f$ be a diagonally non-hyperarithmetical function. If
$f(r(n))=i$ then $\psi_{r(n)}(r(n))\simeq \psi_e(n)\not=i$, which implies $n
\notin P_i$. Therefore, $S=\{n:f(r(n))=1\}$ separates $P_0$ from $P_1$. This
argument is easily relativizable uniformly. The converse direction is also
clear.
\end{proof}

Using a $\Pi^1_1$-norm, one can show $\SiWKL\equivW\PiSep11$ by modifying the
usual proof of the well-known equivalence between ${\sf WKL}$ and $\SiSep01$.

\begin{lemma}\label{lem:wkl-equal-sep}
$\SiWKL \equivW \PiSep11 \equivW \widehat{\SI11\mbox{-}\C_{\mathbf{2}}}$.
\end{lemma}
\begin{proof}
By a straightforward modification of the usual proof of $\SiSep01 \equivW
\widehat{\C_{\mathbf{2}}}$, it is easy to see that $\PiSep11 \equivW
\widehat{\SI11\mbox{-}\C_{\mathbf{2}}}$ holds. It is also clear that
$\PiSep11 \leqW \SiWKL$. Thus, it suffices to show that $\SiWKL \leqW
\PiSep11$.

Given a $\Sigma^1_1$-tree $T\subseteq 2^{<\om}$, let ${\rm Ext}_T \subseteq
2^{<\om}$ be the set of all extendible nodes of $T$. Clearly, its complement
$\neg{\rm Ext}_T= 2^{<\om} \setminus {\rm Ext}_T$ is $\Pi^1_1$, and thus
admits a $\Pi^1_1$-norm $\varphi$ (we need to get $\varphi$ in a uniform way,
but it is straightforward). Consider the $\Pi^1_1$ set $P_i = \{\sigma:
\sigma \fr i <_\varphi \sigma \fr(1-i)\}$ for each $i<2$. Obviously, $P_0
\cap P_1= \emptyset$. We claim that
\[
\sigma\in{\rm Ext}_T\mbox{ and }\sigma\notin P_j\;\Longrightarrow\;\sigma\fr j\in{\rm Ext}_T.
\]
If $\sigma\notin P_j$ then $\sigma\fr j\not<_\varphi\sigma\fr(1-j)$, that is,
either $\varphi(\sigma\fr j)=\infty$ or $\varphi(\sigma\fr (1-j))\leq
\varphi(\sigma\fr j)$ holds. If the former holds then we must have $\sigma\fr
j\in{\rm Ext}_T$. If $\varphi(\sigma\fr j)<\infty$, then we must have
$\varphi(\sigma\fr (1-j))=\infty$ since $\sigma\in{\rm Ext}_T$ implies that
$\sigma\fr i\in{\rm Ext}_T$ for some $i<2$. By the latter condition,
$\infty=\varphi(\sigma\fr (1-j))\leq \varphi(\sigma\fr j)$; hence
$\varphi(\sigma\fr j)$ must be $\infty$. In any case, we have
$\varphi(\sigma\fr j)=\infty$, which means that $\sigma\fr j\in{\rm Ext}_T$.
This verifies the above claim.

Let $S$ be such that $P_0 \subseteq S$ and $S \cap P_1 = \emptyset$. Let
$\sigma_0$ be the empty string, and put $\sigma_{n+1}=\sigma_n\fr
S(\sigma_n)$. Then, by the above claim, we have $\sigma_n\in{\rm Ext}_T$ for
any $n$, and therefore $\bigcup_n\sigma_n\in[T]$. One can easily relativize
this argument uniformly.
\end{proof}

\begin{lemma}
$\SiWKL\leW \widehat{\SI11\mbox{-}\C_\N}$.
\end{lemma}
\begin{proof}
By Lemma \ref{lem:wkl-equal-sep}, we have $\SiWKL \leqW
\widehat{\SI11\mbox{-}\C_\N}$. It remains to show that
$\widehat{\SI11\mbox{-}\C_\N} \nleqW \SiWKL$. It is easy to see that $\SiWKL$
is a cylinder, and hence it suffices to show that
$\widehat{\SI11\mbox{-}\C_\N} \nleqsW \SiWKL$.

We first show the following claim: Let $T\subseteq 2^{<\om}$ be a
$\Sigma^1_1$ tree, and $\Phi$ a Turing functional such that for every $x \in
[T]$, $\Phi^x$ is total. Then there exists a $\Delta^1_1$ function $h : \N
\to \N$ majorizing $n \mapsto \Phi^x(n)$ for every $x \in [T]$.

Let $g:\N\to\N$ be a function such that for any $n$, if $|\sigma|=g(n)$ then
either $\sigma\notin{\rm Ext}_T$ or $\Phi^\sigma(n)\downarrow$. This
condition is clearly $\Pi^1_1$, and by compactness, $g$ is total. Hence, $g$
is a total $\Pi^1_1$ function, and thus actually $\Delta^1_1$. Then define
$h(n)=\max\{\Phi^\sigma(n):|\sigma|=g(n)\mbox{ and
}\Phi^\sigma(n)\downarrow\}$. Clearly $h$ is $\Delta^1_1$ and $\Phi^x(n)\leq
h(n)$ for any $x\in[T]$. This verifies the claim.

Let $(\psi_e)_{e\in\om}$ be a computable enumeration of partial $\Pi^1_1$ functions on $\N$.
Let $S_e$ be the set of all $k$ such that
\[(\forall n\leq e) (\psi_n(e)\downarrow\;\Longrightarrow\;\psi_n(e)<k).\]

Clearly $S_e$ is $\Sigma^1_1$ and cofinite. Then every element of
$S=\prod_eS_e$ dominates all $\Delta^1_1$ functions. If
$\widehat{\Sigma^1_1\mbox{-}\C_\N}\leqsW\SiWKL$ then we must have a
$\Sigma^1_1$-tree $T\subseteq 2^{<\om}$ whose paths compute uniformly an
element of $S$, which is impossible by the above claim.
\end{proof}

Recall that $A\star B$ denotes the sequential composition of $A$ and $B$,
cf.\ \cite{paulybrattka4}, that is, a function attaining the greatest
Weihrauch degree among $\{g\circ f:g\leqW A\mbox{ and }f\leqW B\}$.

\begin{proposition}\label{prop:composition-si11-wkl}
$\SiWKL\star\SiWKL\equivW \SiWKL$.
\end{proposition}

\begin{proof}
This is a modification of the independent choice theorem from
\cite{paulybrattka}. We can assume that the inputs to $\SiWKL\star\SiWKL$ are
a computable function $f$, $z \in \Can$ as well as (relativizable)
$\Sigma^1_1$ trees $S$ and $T$. Then, $\{x\oplus y:x\in [S^z]\mbox{ and }y\in
[T^{f(z,x)}]\}$ is a $\Sigma^1_1$ closed set, and any of its elements is a
solution to $\SiWKL\star\SiWKL$.
\end{proof}

There is a natural principle between $\UC_\Bai$ and $\SiWKL$. Let us define
$\SI11$-weak weak K\"onig's lemma $\SI11\mbox{-}{\sf WWKL}$ as follows: Given
a $\SI11$ set $T\subseteq 2^{<\om}$, if $T$ is an infinite binary tree and if
$[T]$ has a positive measure, then return a path through $T$. This is in
analogy to the usual weak weak K\"onig's lemma, whose Weihrauch degree was
studied in \cite{paulybrattka2,hoelzl,pauly-vitali}.

Note that Hjorth and Nies (see \cite[Chapter 9.2]{Nie09}) showed that there
is a $\Sigma^1_1$-closed set consisting of $\Pi^1_1$-Martin-L\"of random
reals. Indeed, the proof shows that $\Pi^1_1\mbox{-}{\sf MLR}$ is Weihrauch
reducible to $\SI11\mbox{-}{\sf WWKL}$, where $\Pi^1_1\mbox{-}{\sf MLR}$ is a
multivalued functions representing $\Pi^1_1$-Martin-L\"of randomness, which
is introduced in a straightforward manner. We also have ${\sf WKL} \nleqW
\SI11\mbox{-}{\sf WWKL}$ since the Turing upward closure of any nontrivial
separating class has measure zero (cf.~\cite[Theorem 5.3]{JS72}). We show
that, even if we enhance $\SI11\mbox{-}{\sf WWKL}$ by adding a
hyperarithmetical power, its strength is strictly weaker than $\SiWKL$:

\begin{theorem}
$\UC_\Bai \leW \UC_\Bai \star \SI11\mbox{-}{\sf WWKL} \leW \SiWKL$.
\end{theorem}
\begin{proof}
The inequality $\UC_\Bai \leW \UC_\Bai \star \Sigma^1_1\mbox{-}{\sf WWKL}$ is
obvious since no $\Pi^1_1$-Martin-L\"of random real is hyperarithmetic.
Moreover, by Proposition \ref{prop:composition-si11-wkl}, we have $\UC_\Bai
\star \Sigma^1_1\mbox{-}{\sf WWKL} \leqW \SiWKL$. Suppose for the sake of
contradiction that $\SiWKL\leqW\UC_\Bai\star\Sigma^1_1\mbox{-}{\sf WWKL}$.
Then, for any $\Sigma^1_1$ closed set $S$, there are a $\Sigma^1_1$ closed
set $P$ of positive measure and a $\Pi^1_1$ function $f:P\to S$, so that
$f(x)\leq_hx$ for any $x\in P$.

In particular, assume that $S$ is the set of all $\Pi^1_1\mbox{-}{\sf DNC}_2$
functions, and let $P$ and $f$ be as above. It is known that $x$ is
$\Pi^1_1$-random iff $x$ is $\Delta^1_1$-random and $\om_1^{{\rm
CK},x}=\om_1^{\rm CK}$ (see \cite[Theorem 9.3.9]{Nie09}). Since there are
conull many $\Pi^1_1$-random reals, $Q=\{x\in P:\om_1^{{\rm CK},x}=\om_1^{\rm
CK}\}$ also has positive measure. Given $x\in Q$, there is an ordinal
$\alpha<\om_1^{{\rm CK},x}=\om_1^{\rm CK}$ such that $f(x)\leq
_Tx\oplus\emptyset^{(\alpha)}$ (cf.\ \cite[Lemma 4.2]{CY15} and
\cite[Section 2.3.2]{BGM17}). As in \cite[Theorem 5.3]{JS72}, it is easy to
see that the $\emptyset^{(\alpha)}$-Turing upward closure,
${S}_\alpha=\{z:h\leq_Tz\oplus\emptyset^{(\alpha)}\mbox{ for some $h\in
S$}\}$,  of $S$ has measure zero for any computable ordinal $\alpha$. Hence,
$\hat{S}=\bigcup\{S_\alpha:\alpha<\om_1^{\rm CK}\}$ is also null. Our
previous argument shows that $Q\subseteq \hat{S}$, however $\mu(\hat{S})=0$
contradicts $\mu(Q)>0$.
\end{proof}

\begin{question}[\cite{paulybrattka5}]
$\widehat{\Sigma^1_1\mbox{-}\C_\N} \leW \C_\Bai$?
\end{question}


\section{Comparability of well orderings}\label{sec:wellorders}
Two statements which are equivalent to ${\sf ATR}_0$ in the context of
reverse mathematics are comparability of well orderings and weak
comparability of well orderings (\cite[Theorem V.6.8]{SOSOA:Simpson} and
\cite{FriHir:WeakComp}). These involve two kinds of effective witnesses that
one well ordering is shorter than another: strong comparison maps and order
preserving maps.

\begin{definition}\label{def:scm}
If $X,Y\in\WO$ then we say that $f:\N\to \N$ is a \emph{strong comparison
map} between $X$ and $Y$, in symbols $f:X\leq_s Y$, if the following
conditions hold:
	\begin{itemize}
		\item $\forall n (n\notin X\rightarrow f(n)=0)$,
		\item $\forall n,m\in X(n\leq_X m\leftrightarrow f(n)\leq_Y f(m))$,
		\item $\forall n\in X\forall k\in Y(k\leq_Y f(n)\rightarrow  \exists m\in X f(m)=k).$
	\end{itemize}
\end{definition}

In other words, $f$ is an order embedding of $X$ into $Y$ whose image is an initial segment of $Y$.

\begin{definition}[Comparability of well orderings]
Let $\CWO:{\WO\times \WO}\rightarrow \Bai$ be the function that maps any pair
$(X,Y)$ of countable well orderings to the unique $f\in\Bai$ such that
$f:X\leq_s Y$ or $f:{Y+1}\leq_s X$.
\end{definition}

The use of $Y+1$ in the previous definition makes sure that $f$ is unique
even when $X$ and $Y$ are isomorphic.

\begin{definition} If $X, Y\in\LO$ we say that $f:\N\to\N$ is an \emph{order preserving map} between $X$ and $Y$, in symbols $f:X\leq Y$, if the following conditions hold:
	\begin{itemize}
		\item $\forall n(n\notin X\to f(n)=0)$,
		\item $\forall n,m\in X(n\leq_X m\leftrightarrow f(n)\leq_Y f(m))$,
	\end{itemize}
\end{definition}

\begin{definition}[Weak comparability of well orderings]
Let $\WCWO:\WO\times\WO\rightrightarrows\Bai$ be the multivalued function
that maps any pair $(X,Y)$ of countable well orderings to the set
$\set{f\in\Bai}{(f:X\leq Y)\vee (f:Y\leq X)}$.
\end{definition}

The following classifies the Weihrauch degree of comparability of well
orderings:

\begin{theorem}\label{thm:CWOequiv}
$\UC_\Bai\equivsW\CWO$.
\end{theorem}
\begin{proof}
By Lemmata \ref{lemma:cwoupper} and \ref{lemma:CWOlower} below.
\end{proof}

\begin{lemma}\label{lemma:cwoupper}
$\CWO\leqsW\UC_\Bai$.
\end{lemma}
\begin{proof}
If $X,Y\in \WO$, the conjunction of the three conditions in Definition
\ref{def:scm} is a $\PI02$ formula with $X,Y$ and $f$ as free variables. In
particular, a name for the $\PI02$ set $\{f\}=\CWO(X,Y)$ is computable from
$X$ and $Y$. Then, since $\UC_\Bai \equivsW \PI02\mbox{-}\UC_\Bai$ by Theorem
\ref{theo:bigucbaire} and Proposition \ref{pro:closprop}, we can use the
second one to obtain $f$.
\end{proof}

\begin{lemma}\label{lemma:CWOlower}
$\SiSep11\leqsW\CWO$.
\end{lemma}
\begin{proof}
We follow essentially the proof of Theorem V.6.8 in \cite{SOSOA:Simpson}. The
only modification concerns the definition of the well orderings $U$ and $V$,
for which the original proof uses the $\SI11$ bounding principle. 	

So, let $(S_n,T_n)_{n\in \om}$ be a double-sequence of trees in
$\dom(\SiSep11)$. Without loss of generality we assume that for all $n\in\N$,
$S_n$ and $T_n$ are non-empty. We can build the corresponding double-sequence
of linear orderings $(X_n,Y_n)_{n}$ such that, for all $n$, $X_n={\sf
KB}(S_n)$ and $Y_n={\sf KB}(T_n)$. Note that, since
$(S_n,T_n)_{n}\in\dom(\SiSep11)$, we have
\begin{equation}\label{eq:SepCWO1}
\forall n (\WO(X_n)\vee\WO(Y_n)).
\end{equation}
Consider $U=\sum_{n\in\N}(\Q*Y_n)*X_n$, which by (\ref{eq:SepCWO1}) and by
Lemma \ref{lem:ddt}.1 is a well ordering. We claim that the following holds:
\begin{equation}\label{eq:SepCWO2}
\forall X \in \LO\, \forall n (\neg\WO(X_n) \to |X*Y_n|<|U|).
\end{equation}
In fact, let $X \in \LO$ and $n$ be such that $\neg\WO(X_n)$. Then by
(\ref{eq:SepCWO1}) we have $\WO(Y_n)$, which means that $X*Y_n$ is also a
well ordering. Furthermore, by $3$ and $2$ of Lemma \ref{lem:ddt}, we have
$|X*Y_n|\leq |\Q*Y_n|\leq |(\Q*Y_n)*X_n|< |U|$.

For all $n\in\N$, define $Z_n=(U+X_n)*Y_n$. By (\ref{eq:SepCWO2}) and by 1
and 2 of Lemma \ref{lem:ddt} we have, for all $n\in\N$,
\begin{gather}
\label{eq:SepCWO3}\neg\WO(X_n)\to |Z_n|<|U|,\\
\label{eq:SepCWO4}\neg\WO(Y_n)\to |U|< |Z_n|.
\end{gather}
Finally, consider $V=U+\sum_{n\in\N}Z_n$ and define the well orderings
\begin{itemize}
	\item $Z=\sum_{n\in\N}(Z_n+V\cdot\N)$,
	\item $W=\sum_{n\in\N}(V+V\cdot\N)$.
\end{itemize}
Note that all the well orderings we defined so far, in particular $Z$ and
$W$, are computable from the double-sequence $(X_n,Y_n)_n$. In the
construction of $V$ we can also use a special mark for its least element.
Furthermore, we can code $Z$ in such a way that, if $x\in Z_n+V\cdot\N$, for
some $n\in\N$, then we are able to compute whether $x$ belongs to $Z_n$ or to
the first copy of $V$, and in the second case, whether $x$ belongs to the
copy of $U$ contained in $V$. Similar assumptions can be made for the
construction of $W$.

Let now  $f=\CWO(Z,W)$ be the comparing map between $Z$ and $W$. Since
$|Z_n+V\cdot\N|=|V+V\cdot\N|$ for all $n$, we have $|Z|=|W|$ and $f$ is the
isomorphism of $Z$ onto $W$. In particular, for each $n\in\N$, $f$ induces an
isomorphism $f_n$ of $Z_n+V\cdot\N$ onto $V+V\cdot\N$. Define $g \in \Can$ by
$g(n)=0$ if and only if the image of $Z_n$ under $f_n$ is a strict initial
segment of $U$, i.e.\ $|Z_n|<|U|$. This can be done computably by checking
whether $f_n$ maps the first element of the first copy of $V$ in
$Z_n+V\cdot\N$ to $U$ or not. Then, recalling the definition of
$(X_n,Y_n)_n$, if $[S_n]\neq\emptyset$ then $\neg\WO(X_n)$ and, by
(\ref{eq:SepCWO3}), $|Z_n|<|U|$ so that $g(n)=0$. Similarly, if
$[T_n]\neq\emptyset$ then, by (\ref{eq:SepCWO4}), $|U|\leq|Z_n|$ so that
$g(n)=1$.
\end{proof}

The Weihrauch degree of weak comparability of well orderings, however, has
eluded our classification attempts:

\begin{question}\label{question:wcwo}
Does $\WCWO \equivW \UC_\Bai$?
\end{question}

Recently, Jun Le Goh \cite{Goh} obtained a positive answer to our question.


\section{The one-sided versions of $\mathsf{PTT}$ and open determinacy}\label{sec:onesided}
Both the perfect tree theorem and open determinacy have at its core a
disjunction $A \vee B$ which is {\bf not} to be read constructively. A
typical approach to formulate these as computational tasks is to view these
as implications $\neg A \Rightarrow B$ or $\neg B \Rightarrow A$. In this
section, we explore these variants.

Recall that a tree is perfect if every node has at least two incomparable
extensions. In particular, every perfect tree is pruned. The perfect tree
theorem states that every tree with uncountably many paths has a perfect
subtree and leads to the following two problems: The first problem is given a
closed set $A$ which has no perfect subset (that simply means that $A$ is
countable), and has to show its countability, that is, to enumerate all
elements of $A$. We consider two variants of this task, depending on what
exactly is meant by \emph{listing}. The weak version contains no information
about the cardinality, the strong version does. The second problem is more
direct: it asks to find a perfect subset of a given tree with uncountably
many paths.

\begin{definition}
$\wList : \subseteq \mathcal{A}(\Bai) \mto (\Bai)^\omega$ maps a countable
set $A$ to some $\langle b_0p_0,b_1p_1,\ldots \rangle$ such that $A = \{p_i
\mid b_i = 1\}$. $\List : \subseteq \mathcal{A}(\Bai) \mto (\Bai)^\omega$
maps a countable set $A$ to some $n \langle p_0,p_1,\ldots \rangle$ such that
either $n = 0$, $p_i \neq p_j$ for $i \neq j$ and $A = \{p_i \mid i \in
\N\}$; or $n
> 0$, $|A| = n - 1$ and $A = \{p_i \mid i < n - 1\}$.
\end{definition}

\begin{definition}
$\PTT1 : \subseteq \Tr \mto \Tr$ maps $T$ such that $[T]$ is uncountable to
some perfect $T' \subseteq T$.
\end{definition}

We start by reporting a result originating from discussion during the
Dagstuhl seminar on Weihrauch reducibility \cite{pauly-dagstuhl}, in
particular including a contribution by Brattka:

\begin{proposition}
\label{prop:ptt1} $\PTT1 \equivW \C_\Bai$.
\begin{proof}
For $\C_\Bai \leqW \PTT1$, note that from $A \in \mathcal{A}(\Bai)$ we can
compute a tree $T$ such that $[T] =  A \times \Bai $. If $A$ is non-empty,
then $[T]$ is uncountable. Given some perfect subtree $T'$ of $T$, we can
compute a path through $T'$ and hence through $T$. By projecting, we obtain a
point in $A$.

For $\PTT1 \leqW \C_\Bai$, call a function $\lambda : \N^{<\N} \to \N$ a
modulus of perfectness for $T$, if $v \in T$ implies that there are
incomparable $u, w \in [0,\lambda(v)]^{\lambda(v)}$ with $vu, vw\in T$. A
non-empty tree has a modulus of perfectness iff it is perfect, and given $T$
the set \[\{( T',\lambda ) \in \Tr \times \N^{(\N^{<\N})} \mid \emptyset
\not= T' \subseteq T \wedge \lambda\  \textnormal{ is a modulus of
perfectness for } T'\}\] is closed, and non-empty for $[T]$ uncountable by
the perfect tree theorem. Taking into account that $\Tr \times
\N^{(\N^{<\N})}$ is computably isomorphic to $\Bai$, we can thus apply
$\C_\Bai$ and project to obtain a perfect subtree of $T$.
\end{proof}
\end{proposition}

\subsection{Listing the points in a countable set}
We now examine the strength of the contrapositive of the perfect tree theorem
$\PTT1$, which is $\List$ in our setting as explained above.

\begin{theorem}
\label{theo:list} $\wList \equivW \List \equivW \UC_\Bai$.
\end{theorem}

The main ingredient of our proof is a variant of the Cantor-Bendixson
decomposition, designed in such a way that it can be carried out in a Borel
way. This modified version works as the usual one for countable sets, but can
differ for uncountable ones\footnote{Kreisel has shown that computable $A \in
\mathcal{A}(\Bai)$ may have uncomputable Cantor-Bendixson rank
\cite{kreisel}. As any total function from $\Bai$ into the countable ordinals
that is effectively Borel is dominated by a computable function (the Spector
$\Sigma^1_1$-boundedness principle, cf.\ \cite{pauly-ordinals}), this implies
that the Cantor-Bendixson decomposition cannot be done in a Borel way.}. If
$u$ and $w$ are finite words on $\N$, $u \sqsubseteq w$ means that $u$ is a
prefix of $w$.

\begin{definition}
A \emph{one-step mCB-certificate} of $A \in \mathcal{A}(\Bai)$ consists of
\begin{enumerate}[label=(\alph*)]
\item A prefix-independent\footnote{Meaning that $w_i \sqsubset w_j$ never
    holds.} sequence $(w_i)_{i \in \N}$ of finite words ordered in a
    canonical way,
\item A sequence of bits $(b_i)_{i \in \N}$ which are not all $0$,
\item A sequence of points $(p_i)_{i \in \N}$
\end{enumerate}
subject to the following constraints:
\begin{enumerate}
\item If $b_i = 1$, then $p_i \in A \cap w_i\Bai$.
\item If $b_i = 0$, then $\forall p \in \mathrm{HYP}(A) \ p \notin A \cap
    w_i\Bai$ and $p_i = 0^\omega$.
\item $\forall p, q \in \mathrm{HYP}(A) \ \left ( p \in A \cap w_i\Bai
    \wedge q \in A \cap w_i \Bai  \Rightarrow p = q \right )$.
\item If $w_i \not\sqsubseteq w$ for all $i \in \N$, then $\exists p, q \in
    A \cap w\Bai \ p \neq q$.
\end{enumerate}
For a one-step mCB-certificate for $A$, its \emph{residue} is $A \setminus
\bigcup_{i \in \N} w_i\Bai$.
\end{definition}

\begin{definition}
A \emph{global mCB-certificate} for $A \in \mathcal{A}(\Bai)$ is indexed by
some initial $I \subseteq \N$ (which may be empty). It consists of a sequence
$(c_i)_{i \in I}$ of one-step mCB-certificates such that there exists a
linear ordering $\mathalpha{\sqsubset} \subseteq I \times I$ with minimum $0$
(if non-empty), such that $c_0$ is a one-step mCB-certificate for $A$, for
each $n \in I \setminus \{0\}$, $c_n$ is an mCB-certificate for $\bigcap_{i
\sqsubset n} A_i$, where $A_i$ is the residue of $c_i$; and $\forall p \in
\mathrm{HYP}(A) \ p \notin A \cap \bigcap_{i \in I} A_i$.
\end{definition}

\begin{lemma}\label{lemma:globalmcb}
The set of global mCB-certificates of $A$ is uniformly $\Sigma^1_1$ in $A$.
\begin{proof}
This is almost immediate from the definition, besides the quantification over $\mathrm{HYP}$. That this is unproblematic follows from Kleene's $\mathrm{HYP}$-quantification theorem \cite{kleene4,kleene5} (the converse of the Spector-Gandy theorem).
\end{proof}
\end{lemma}

\begin{lemma}
For non-empty non-perfect $A \in \mathcal{A}(\Bai)$, $A$ has a one-step
mCB-certificate such that its residue is equal to its Cantor-Bendixson
derivative. If all points in $A$ are hyperarithmetical relative to $A$, then
$A$ has a unique one-step mCB-certificate.
\begin{proof}
Let $(q_j)$ be the finite or infinite list of isolated points in $A$, and let
$(u_j)$ be the shortest prefix such that $A \cap u_j\Bai = \{q_j\}$. It
follows from Corollary \ref{corr:uc-hyppoints} applied to $A \cap u_j\Bai$
that each $q_j$ is hyperarithmetical relative to $A$. Let $(v_k)$ be the list
of shortest prefixes such that $A \cap v_k\Bai = \emptyset$, excluding those
extending some $u_j$. Now the sequence $(w_i)$ is obtained such that $\{w_i\}
= \{u_j\} \cup \{v_k\}$, subject to the canonical ordering condition. If $w_i
= v_k$, then $b_i = 0$ and $p_i = 0^\omega$, if $w_i = u_j$ then $b_i = 1$
and $p_i = q_j$.

It is immediate that the construction satisfies Conditions (1,2,3,4) and that
the residue sees exactly the isolated points removed, i.e.~is the
Cantor-Bendixson derivative of $A$. It remains to argue that the
mCB-certificate constructed as such is unique if all points in $A$ are
hyperarithmetical relative to $A$ (this is a classic result, of course). As
the choice of $b_i$ and $p_i$ was uniquely determined by the sequence
$(w_i)$, we only need to prove that there is no alternative sequence
$(w'_i)$. As no $w_i$ can satisfy the conclusion of Condition (4), we know
that for each $w_i$ there exists some $w'_{i'}$ with $w'_{i'} \sqsubseteq
w_i$.

Assume that $w'_{i'} \sqsubset w_i$ for some $i$. If $b_i = 1$, then $w_i$
was chosen minimal under the constraint that $A \cap w_i\Bai$ is a singleton,
$A \cap w'_{i'}$ contains at least two points, which are both
hyperarithmetical. Hence, $w'_{i'}$ fails Condition (3). If $b_i = 0$, then
$w'_{i'}\Bai \cap A = \emptyset$ contradicts the choice of $v_k$ as shortest
prefix, $|w'_{i'}\Bai \cap A| = 1$ contradicts the choice of $u_j$ as
shortest prefix of an isolated point in $A$, and $|w'_{i'}\Bai \cap A| \geq
2$ again violates Condition (3). Hence we know that all $(w_i)$ must appear
as some $(w'_{i'})$.

Assume that there is some $w$ occurring as a $w'_{i'}$ but not as a $w_i$. As
the $(w'_{i'})$ are prefix-free, $w$ is not an extension of some $w_i$.
Hence, Condition (4) for the $(w_i)$ implies that $|A \cap w\Bai| \geq 2$.
But as all points in $A$ are hyperarithmetical, this shows that neither the
conclusion of Condition (2) nor that of Condition (3) can be satisfied for
$w'_{i'} = w$, and we have obtained the desired contradiction.
\end{proof}
\end{lemma}

\begin{corollary}
\label{corr:mcb} If $A \in \mathcal{A}(\Bai)$ is countable, then $A$ has a
unique global mCB-certificate, the $p_i$ for $b_i = 1$ occurring in some
one-step mCB-certificate list all points in $A$, and the order type of the
implied linear ordering is the Cantor-Bendixson rank of $A$ plus $1$.
\end{corollary}

\begin{proof}[Proof of Theorem \ref{theo:list}]
That $\UC_\Bai \leqW \wList$ is simple: Any instance of the former is an
instance of the latter, and from a list repeating a single element, we can
recover that element. For the other direction, we show $\wList \leqW \SiUC$
instead and invoke Theorem \ref{theo:bigucbaire}. By Lemma
\ref{lemma:globalmcb} the set of global mCB-certificates of $A \in
\mathcal{A}(\Bai)$ is computable as a $\Sigma^1_1$-set from $A$, and by
Corollary \ref{corr:mcb} this is a singleton for countable $A$. We can
distinguish whether the global mCB-certificate uses an empty or non-empty
linear order. In the former case, the set is empty, and in the latter case,
we can compute a list of all points in $A$.

Again, $\wList \leqW \List$ is trivial. For the reverse direction, we observe
that $\List \leqW \UC_\Bai \star \wList$, since $\UC_\Bai$ more than suffices
to extract the required additional information from an unstructured list. We
then use the preceding result and $\UC_\Bai \equivW \UC_\Bai \star \UC_\Bai$
from \cite{paulybrattka}.
\end{proof}

Regarding the non-uniform aspect, it is known that every countable $\Pi^0_1$
(indeed $\Sigma^1_1$) set $A\subseteq\Bai$ consists only of hyperarithmetical
elements (\cite[Theorem III.6.2]{sacks2}). Theorem \ref{theo:list} concludes
that every countable $\Pi^0_1$ set $A\subseteq\Bai$ admits a
hyperarithmetical enumeration. Combining Proposition \ref{prop:ptt1} (and
Gandy's basis theorem \cite[Corollary III.1.5]{sacks2}) and Theorem
\ref{theo:list}, we indeed get the following:

\begin{corollary}\label{cor:PTT-nonuniform}
For any computable tree $T\subseteq\om^{<\om}$, either $T$ has a hyperlow
perfect subtree or there is a hyperarithmetical enumeration of all infinite
paths through $T$.
\end{corollary}

\subsubsection*{Listing on Cantor space}
We have seen that for subsets of Baire space, it makes no difference whether
we intend to list all points of a countable set or all points of a finite
set. We briefly explore the corresponding versions for Cantor space. Let
$\List_{\Can,<\omega} : \subseteq \mathcal{A}(\Can) \mto (\Can)^*$ denote the
problem to produce a tuple of all elements of a finite closed subset of
$\Can$ (i.e.~$(p_0,\ldots,p_{n-1}) \in \List_{\Can,<\omega}(A)$ iff $A =
\{p_i \mid i < n\}$). Let $\wList_{\Can,\leq \omega} : \subseteq
\mathcal{A}(\Can) \mto (\Can)^\omega$ denote the problem to list all elements
of a non-empty countable closed subset of $\Can$ (i.e.~$(p_i)_{i \in
\mathbb{N}} \in \wList_{\Can,\leq \omega}(A)$ iff $\{p_i \mid i \in
\mathbb{N}\} = A$). Note that $\List_{\Can,<\omega}$ is not a restriction of
$\wList_{\Can,\leq \omega}$, since finite tuples and lists with finite range
have distinct properties. We will in fact show in Corollary \ref{List-wList}
that these two multivalued functions are incomparable with respect to
Weihrauch reducibility.

\begin{proposition}
\label{prop:finitelist} $\List_{\Can,<\omega} \equivW \PI02\mbox{-}\C_\N$.
\begin{proof}
To see that $\List_{\Can,<\omega} \leqW \PI02\mbox{-}\C_\N$, note that we can
guess a finite partition of $\Can$ into clopens $A_0,\ldots,A_n$ such that
$|A \cap A_i| = 1$ for input $A$ and any $i$. Verifying a correct partition
is $\PI02$ (because $A \cap A_i \neq \emptyset$ and $|A \cap A_i| \leq 1$ are
respectively a $\PI01$ and a $\PI02$ condition), and given a correct
partition, we can compute the listing since $\UC_\Can$ is computable.

For the other direction, note that we can view $\PI02\mbox{-}\C_\N$ as the
following task: Given $(p_0,p_1,\ldots ) \in (\Can)^\omega$ with the promise
that if $|\{j \mid p_i(j) = 1\}| = \infty$ then $|\{j \mid p_{i+1}(j) = 1\}|
= \infty$, and that there exists some $i$ with $|\{j \mid p_i(j) = 1\}| =
\infty$, find such an $i$ (for details, see \cite{paulybrattka5}). We now
construct $A \in \Can$ as follows: For each $i$, keep track of an auxiliary
variable $k_i$, which is initially $0$. Start enumerating all $0^{\langle
i,k\rangle}1$ into the complement of $A$ except the $0^{\langle i,
k_i\rangle}1$. Also enumerate all $0^l1^s0$. Whenever we read another $1$ in
$p_i$, we do enumerate $0^{\langle i, k_i\rangle}1$, and set the new $k_i$ to
be the least $k$ such that $0^{\langle i, k\rangle}1$ has not been enumerated
yet.

Whenever $|\{j \mid p_i(j) = 1\}| < \infty$ for some $i$, then $k_i$ will eventually remain constant. The resulting set $A$ will be of the form $\{0^\omega\} \cup \{0^{\langle i, k_i\rangle}1^\omega \mid i \in I\}$ where $I$ is the finite set of non-solutions. Having a finite listing of $A$ lets us easily pick some solution.
\end{proof}
\end{proposition}

As a corollary one can see that every finite $\Pi^0_1$ subset of $2^\N$
admits a computable listing uniformly in $\mathbf{0}''$, and the complexity
$\mathbf{0}''$ is optimal: If a function $f$ sends an index (i.e.\ a G\"{o}del
number) of a $\Pi^0_1$ set $P\subseteq 2^\N$ to an index of a computable
listing of elements of $P$ whenever $P$ is finite, then $f$ must compute
$\mathbf{0}''$.

\begin{proposition}
\label{prop:cantorlist} $\wList_{\Can,\leq\omega} \equivW
\widehat{\wList_{\Can,\leq\omega}} \leqW \UC_\Bai \equivW \PI02\mbox{-}\C_\N
\star \wList_{\Can,\leq\omega}$.
\begin{proof}
To note that $\wList_{\Can,\leq\omega}$ is parallelizable, observe that we
can effectively join countably many trees along a comb, and the set of paths
of the result is essentially the disjoint union of the original paths. The
second reduction follows from the obvious embedding of $\Can$ into $\Bai$ as
a closed set and Theorem \ref{theo:list}. For the third reduction, note that
we can embed $\Bai$ as a $\Pi^0_2$-subspace $B$ into $\Can$ such that $\Can
\setminus B$ is countable. Given some singleton $A \in \mathcal{A}(\Bai)$, we
can compute some countable $\bar{A} \in \mathcal{A}(\Can)$ such that $\bar{A}
\cap B$ is the image of $A$ under that embedding. If we have a list of all
points in $\bar{A}$, we can then use $\Pi^0_2\mbox{-}\C_\N$ to pick the one
in $B$. That the third reduction is an equivalence follows from the second,
the observation that $\PI02\mbox{-}\C_\N \leqW \UC_\Bai$ and $\UC_\Bai \star
\UC_\Bai \equivW \UC_\Bai$ (cf.~\cite{paulybrattka}).
\end{proof}
\end{proposition}

\begin{proposition}
\label{prop:lim-list-strength-aux} ${\sf lim} \leqW
\wList_{\Can,\leq\omega}$.
\begin{proof}
Consider the map $\id : \mathcal{A}(\N) \to \mathcal{O}(\N)$ translating an
enumeration of a complement of a set to an enumeration of the set. Studied
under the name $\mathrm{EC}$ in \cite{stein}, it is known to be equivalent to
${\sf lim}$. Now from $A \in \mathcal{A}(\N)$ we can compute $\{0^\omega\}
\cup \{0^n1^\omega \mid n \in A\} \in \mathcal{A}(\Can)$. From any list of
the elements of the latter set, we can then compute $A \in \mathcal{O}(\N)$.
\end{proof}
\end{proposition}

\begin{proposition}\label{prop:lim-list-strength}
The following are equivalent for single-valued $f: \subseteq \mathbf{X} \to
\Bai$ where $\mathbf{X}$ is a represented space:
\begin{enumerate}
\item $f \leqW {\sf lim}$;
\item $f \leqW \wList_{\Can,\leq\omega}$.
\end{enumerate}
\begin{proof}
Proposition \ref{prop:lim-list-strength-aux} entails that $1.$ implies $2.$

To see that $2.$ implies $1.$, consider some single-valued $f : \subseteq
\Bai \to \Bai$ with $f \leqW \wList_{\Can,\leq\omega}$. So from any $p \in
\dom(f)$, we can compute some countable $A_p \in \mathcal{A}(\Can)$, and from
any enumeration of the points in $A_p$ together with $p$ we can compute
$f(p)$ via some computable $K$. We will argue that having access to a pruned
tree $T$ with $[T] = A_p$ suffices to compute $f(p)$, and note that pruning a
binary tree is equivalent to ${\sf lim}$ (see
e.g.~\cite{pauly-nobrega-arxiv}). Let us assume that there are prefixes $w_0,
\ldots, w_n$ in the pruned tree such that $K$ upon reading $p$ and
$w_0,\ldots,w_n$ outputs some prefix $w$. Then there is some enumeration
$q_0,q_1,\ldots$ of points in $A_p$ such that $w_0,\ldots,w_n$ are prefixes
of $q_0,\ldots,q_n$, hence $w$ is a prefix of $f(p)$. Conversely, for any
fixed enumeration $q_0,q_1,\ldots$ of points in $A_p$ and desired prefix
length $m$ of $f(p)$ there is some $k \in \N$ such that $K$ outputs
$f(p)_{\leq m}$ after having read no more than the $k$-length prefixes of
$q_i$ for $i \leq k$. Moreover, each $(q_i)_{\leq k}$ occurs in the pruned
tree $T$. Thus, having access to $T$ lets us compute longer and longer
prefixes of $f(p)$, and since $f$ is single-valued, this suffices to compute
$f(p)$.
\end{proof}
\end{proposition}

In particular, $A\subseteq\N$ is computable from all listings of some countable $\Pi^0_1$ set $P\subseteq 2^\N$ iff $A$ is $\mathbf{0}'$-computable. On the other hand, there is no computable ordinal $\alpha$ such that $\mathbf{0}^{(\alpha)}$ computes a listing of any countable $\Pi^0_1$ subset of $2^\N$.

\begin{corollary}\label{List-wList}
$\List_{\Can,<\omega} \nleqW \wList_{\Can,\leq\omega}$ and
$\wList_{\Can,\leq\omega} \nleqW \List_{\Can,<\omega}$.
\end{corollary}

\begin{proof}
For the first claim, it is known that $\PI02$-$\C_\N \equivW \PI02$-$\UC_\N$
\cite{paulybrattka5}. (Sketch: Take $(p_i)_{i\in\N}$ as in Proposition
\ref{prop:finitelist}, and then put $\hat{p}_{i,s}(n)=1$ iff $p_{i}(n)=1$ and
$p_j(t)=0$ for all $j<i$ and $s\leq t< n$. It is easy to see that there is a
unique $i,s$ such that $|\{n\mid \hat{p}_{i,s}(n)=1\}|=\infty$, and then
$|\{n\mid p_i(n)=1\}|=\infty$.) Then observe that $\PI02$-$\UC_\N$ is
single-valued, and that ${\sf lim}$ is $\Sigma^0_2$-computable while
$\PI02$-$\C_\N$ is not. The claim then follows by Proposition \ref{prop:lim-list-strength}.

The second claim follows from the observation that any solution of a
(computable) instance of $\PI02$-$\C_\N$ must be computable, while ${\sf
lim}$ has computable instances without computable solutions.
\end{proof}

\begin{corollary}
$\wList_{\Can,\leq\omega} \leW \wList_{\Can,\leq\omega} \star
\wList_{\Can,\leq\omega} \star \wList_{\Can,\leq\omega} \equivW \UC_\Bai$.
\end{corollary}

\begin{proof}
In Proposition \ref{prop:lim-list-strength-aux} we have shown that ${\sf
lim}\leqW\List_{\Can,\leq\omega}$, which implies
$\Pi^0_2\mbox{-}\C_\N\leqW{\sf lim} \star {\sf lim}
\leqW\wList_{\Can,\leq\omega}\star \wList_{\Can,\leq\omega}$; hence the
assertion follows from Proposition \ref{prop:cantorlist} and $\UC_\Bai \star
\UC_\Bai \equivW \UC_\Bai$. The strictness follows from Proposition
\ref{prop:lim-list-strength} since $\UC_\Bai$ is single-valued and $\UC_\Bai
\nleqW {\sf lim}$.
\end{proof}

\begin{question}
Does $\wList_{\Can,\leq\omega} \star \wList_{\Can,\leq\omega} \equivW
\UC_\Bai$ hold?
\end{question}

The feature that $\wList_{\Can,\leq\omega}$ is not closed under composition
itself, but that the hierarchy of more and more compositions stabilizes at a
finite level, seems surprising for a \emph{natural} degree. A similar
observation was made before regarding the degree of finding Nash equilibria
in bimatrix games \cite{pauly-kihara2-mfcs}.

\subsection{Finding winning strategies}

We now move on to the complexity of finding winning strategies in open
Gale-Stewart games. In formulating the corresponding multivalued functions,
we implicitly code strategies in sequential games into Baire space elements.

\begin{definition}
$\FindWSS :\subseteq \mathcal{O}(\Bai) \mto \Bai$ ($\FindWSP : \subseteq
\mathcal{O}(\Bai) \mto \Bai$) maps an open game where Player 2 (Player 1) has
no winning strategy to a winning strategy for Player 1 (Player 2). Likewise,
$\FindWSD$ maps a clopen game where Player 2 has no winning strategy to a
winning strategy for Player 1. Here a name for a clopen set consists of two
names for open sets which are one the complement of the other.
\end{definition}

On the one hand, the difficulty of finding a winning strategy for a closed
player is the same as the closed choice on Baire space.

\begin{proposition}
\label{prop:findwspi} $\FindWSP \equivW \C_\Bai$.
\begin{proof}
For $\C_\Bai \leqW \FindWSP$, note that we can turn any $A \in
\mathcal{A}(\Bai)$ into a \SI01 game where Player $1$'s moves do not matter,
and Player $2$ wins iff his moves form a point $p \in A$.

For $\FindWSP \leqW \C_\Bai$, note that given a Player $2$ strategy $\tau$
and the \SI01 winning condition $W \subseteq \Bai$ we can compute a tree
$T_{W,\tau}$ describing the options available to Player $1$: Essentially, the
strategies $\sigma$ winning against $\tau$ correspond to finite paths in
$T_{W,\tau}$ ending in a leaf, whereas strategies $\sigma'$ losing against
$\tau$ correspond to infinite paths through $T_{W,\tau}$. Thus, $\tau$ is a
winning strategy for Player $2$ iff $T_{W,\tau}$ is a pruned tree, i.e.~a
tree without any leaves. Let $\lambda : \N^* \to \N$ be a witness of
prunedness of $T$ iff $\forall v \in T \ v\lambda(v) \in T$. If Player $2$
has a winning strategy for the game $W$, then the set
\[\{(\tau, \lambda) \mid \lambda\ \textnormal{ is a witness of prunedness for } T_{W,\tau}\}\]
is a non-empty closed set computable from $W$, and projecting a member of it yields a winning strategy for Player $2$.
\end{proof}
\end{proposition}

On the other hand, the difficulty of finding a winning strategy for a open/clopen player is the same as the unique choice on Baire space. In the case of clopen games, we even get full determinacy defined as follows:

\begin{definition}
$\DetD : \DE01(\Bai) \mto \Bai \times \Bai$ maps a clopen game $W$ to a pair
of strategies $\sigma$, $\tau$ such that either $\sigma$ is winning for
Player $1$ or $\tau$ is winning for Player $2$ (i.e.~a Nash equilibrium).
\end{definition}

\begin{theorem}\label{theo:findWSSigma}
$\FindWSD \equivW \DetD \equivW \FindWSS \equivW \UC_\Bai$.
\end{theorem}

We will prove Theorem \ref{theo:findWSSigma} using the following lemmata.

\begin{lemma}
\label{lemma:findWSSigma} $\FindWSS \leqW \SiUC$.
\begin{proof}
Let $T$ be a tree describing the complement of some open set, the payoff for
Player $1$. Fix some strategy $\sigma$ of Player $1$. We understand this to
prescribe the action even at positions made impossible by $\sigma$ itself.
For any $v \in \N^*$ where Player $1$ moves, consider the trees $T_i^v$
describing the options available to Player $2$ if the game starts at $v$,
Player $1$ plays $i$ and otherwise follows $\sigma$. $\sigma$ is a winning
strategy iff for any $v$ compatible with $\sigma$ we find that
$T_{\sigma(v)}^v$ is well-founded. Only $\SiUC$ is available here while a lot
of strategies may exist. We overcome this difficulty by considering the
optimal strategy, that is, the one that minimizes the rank of
$T_{\sigma(v)}^v$.

Let $v$ be a position where Player $1$ moves. A certificate of optimality for
$\sigma$ at $v$ describes maps preserving $\sqsubset$ from $T^v_{\sigma(v)}$
to $T^v_i \setminus \{\lambda\}$ (here $\lambda$ denotes the empty sequence)
for every $i < \sigma(v)$, and maps preserving $\sqsubset$ from
$T^v_{\sigma(v)}$ to $T^v_j$ for every $j > \sigma(v)$. The set of strategies
$\sigma$ and corresponding certificates of optimality for all positions is a
closed set computable from the game.

If we fix partial strategies of all proper extensions of $v$ such that Player
$1$ can win from $v$, then there is a unique action of Player $1$ at $v$ such
that extending the strategy to $v$ admits a certificate of optimality. It
follows that if Player $1$ has a winning strategy, then there is a unique
strategy admitting a certificate of optimality at all compatible positions;
and this strategy is winning. We can compute this using
$\SiUC$.
\end{proof}
\end{lemma}

\begin{corollary}
\label{corr:findwssigma} $\FindWSS \leqW \UC_\Bai$.
\begin{proof}
By Lemma \ref{lemma:findWSSigma} and Theorem \ref{theo:bigucbaire}.
\end{proof}
\end{corollary}

\begin{lemma}
\label{lemma:fixwinner} $\DetD \leqW \FindWSD$.
\begin{proof}
Given a $\DE01$-game $G$, we can compute the derived $\DE01$-game $G'$ where
the first player can decide whether to play $G$ as Player $1$, or as Player
$2$, and then proceed a play of a chosen side. Thus, Player $1$ can
definitely win $G'$, and a winning strategy of Player $1$ in $G'$ tells us
who wins $G$ and how.
\end{proof}
\end{lemma}

\begin{lemma}
\label{lemma:findWSparallelizable} $\widehat{\FindWSD} \leqW \FindWSD$.
\begin{proof}
Given a sequence $G_0, G_1, \ldots$ of $\DE01$-games all won by Player $1$,
we combine them into a single \DE01 game where Player $2$ first chooses $n$,
and then the players play $G_n$. Player $1$ wins the combined game, and any
winning strategy in that game yields in the obvious way winning strategies
for every $G_i$.
\end{proof}
\end{lemma}

Let $\mathbb{S}_\mathcal{B}$ denote the space of Borel-truth values
(cf.~\cite{pauly-gregoriades,pauly-ordinals}). Roughly speaking, if $p$ is a
Borel code of a Borel subset $A$ of the singleton space $\{\bullet\}$, then
we think of $p$ as a name of $\top$ ($\bot$, resp.)\ iff $A\not=\emptyset$
($A=\emptyset$, resp.); if $p$ is not a Borel code, $p$ is not in the domain
of the representation.

\begin{lemma}
\label{lemma:boreltodelta} $\left ( \id : \mathbb{S}_\mathcal{B} \to
\mathbf{2} \right ) \leqW \DetD$.
\begin{proof}
A Borel code can be viewed as a well-founded tree whose even-levels
(odd-levels, resp.)\ consist of $\exists$-vertices ($\forall$-vertices,
resp.)\ and leaves are labeled by either $\top$ or $\bot$ (corresponding to
either $\{\bullet\}$ or $\emptyset$) \cite{pauly-gregoriades,pauly-ordinals}.
We can turn a $\mathbb{S}_\mathcal{B}$-name into a $\DE01$-game by letting
Player 1 control the $\exists$-vertices, Player 2 the $\forall$-vertices,
make the $\top$-leaves winning for Player $1$ and the $\bot$-leaves losing.
Then Player $1$ has a winning strategy iff the value of the root is $\top$.
Given a Nash equilibrium $(\sigma, \tau)$ we can compute the leaf reached by
the induced play, and find it to be equal to the truth value of the root.
\end{proof}
\end{lemma}

\begin{proof}[Proof of Theorem \ref{theo:findWSSigma}]
As shown in \cite[Theorem 80]{pauly-ordinals}, $\UC_\Bai \leqW \widehat{\left ( \id :
\mathbb{S}_\mathcal{B} \to \mathbf{2} \right )}$. By Lemma
\ref{lemma:boreltodelta}, the latter is reducible to $\widehat{\DetD}$. This
is reducible to $\widehat{\FindWSD}$ by Lemma \ref{lemma:fixwinner}, which in
turn reduces to $\FindWSD$ by Lemma \ref{lemma:findWSparallelizable}.
$\FindWSD \leqW \DetD$ is trivial, and so is $\FindWSD \leqW \FindWSS$.
$\FindWSS \leqW \UC_\Bai$ follows by Corollary \ref{corr:findwssigma}.
\end{proof}

As in the case of the perfect tree theorem (Corollary \ref{cor:PTT-nonuniform}), the results in this section can be viewed as a refinement of the following known result \cite{Blass:ComWinStr}:

\begin{corollary}
For any open game, either the open player has a hyperarithmetical winning strategy or the closed player has a hyperlow winning strategy.
\end{corollary}

\section{The two-sided versions of $\mathsf{PTT}$ and open determinacy}
\label{sec:twosided}
Rather than demanding a promise about the case of the theorem we are in, we could alternatively consider the task completely uniformly. As distinguishing the two cases is a $\Pi^1_1$-complete question (cf.\ the well-known equation $\Game\Sigma^0_1=\Pi^1_1$), the fully uniform task should {\bf not} include the information in which case we are. A priori, since we considered two versions of listing, we also have the two corresponding version of the two-sided perfect tree theorem. We are left with the following formulations:

\begin{definition}
$\wPTT2 : \Tr \mto \Tr \times \Bai$ has $(T',\langle
b_0p_0,b_1p_1,b_2p_2,\ldots \rangle) \in \wPTT2(T)$ iff one of the following
holds:
\begin{itemize}
  \item $T'$ is a perfect subtree of $T$;
  \item $[T] = \{p_i \mid b_i \neq 0\}$
\end{itemize}
\end{definition}

\begin{definition}
$\PTT2 : \Tr \mto \Tr \times \Bai$ has $(T',n\langle p_0,p_1,p_2,\ldots
\rangle) \in \PTT2(T)$ iff one of the following holds:
\begin{itemize}
  \item $T'$ is a perfect subtree of $T$;
  \item $n = 0$, $p_i \neq p_j$ for $i \neq j$ and $[T] = \{p_i \mid i \in
      \N\}$;
  \item $n > 0$, $|[T]| = n - 1$ and $[T] = \{p_i \mid i < n - 1\}$.
\end{itemize}
\end{definition}

\begin{definition}
$\DetS : \mathcal{O}(\Bai) \mto \Bai \times \Bai$ maps an open game $W$ to a
pair of strategies $\sigma$, $\tau$ such that either $\sigma$ is winning for
Player $1$ or $\tau$ is winning for Player $2$.
\end{definition}

These variants are strictly harder than the non-uniform ones (which are
Weihrauch reducible to $\C_\Bai$ by the results of Section
\ref{sec:onesided}). To see that, let $\chiP: \Bai \to \mathbf{2}$ be the
characteristic function of a $\Pi^1_1$-complete set. Since the single-valued
functions between computable Polish spaces which are Weihrauch reducible to
$\C_\Bai$ are exactly those that are effectively Borel measurable
(\cite[Theorem 7.7]{paulybrattka}), and $\chiP$ is not such, we have $\chiP
\nleqW \C_\Bai$.

\begin{observation}
$\chiP \leqW \lpo' \star \wPTT2$ and $\chiP \leqW \lpo \star \DetS$.
\begin{proof}
Deciding whether $[T]$ is uncountable and who wins a $\SI01$-game are
$\PI11$/$\SI11$-complete decision problems. Given trees $T'$ and $T$, we can
use $\lpo'$ to decide whether or not $T'$ is a perfect subtree of $T$. Given
a Nash equilibrium $(\sigma,\tau)$ of a $\SI01$-game, we can compute the
induced play and then use $\lpo$ to decide who wins that play -- and this is
the same player that has a winning strategy in the game.
\end{proof}
\end{observation}

\begin{corollary}\label{cor:two-sided-det-ptt}
$\C_\Bai \leW \wPTT2 \leqW \PTT2$ and $\C_\Bai \leW \DetS$.
\begin{proof}
Using the fact that $\C_\Bai$ is closed under composition \cite[Corollary
7.6]{paulybrattka} we have $\chiP \nleqW \C_\Bai \equivW \lpo \star \C_\Bai
\equivW \lpo' \star \C_\Bai$.
\end{proof}
\end{corollary}

In particular, we find that $\FindWSS \leW \DetS$ and $\FindWSP \leW \DetS$.
Thus, knowing who wins a $\SI01$-game makes it strictly easier to find a Nash
equilibrium. This is in contrast to $\DE01$-games (as seen in Theorem
\ref{theo:findWSSigma}), as well as to games on Cantor space with winning
sets in the difference hierarchy over $\SI01$ (cf.~\cite{paulyleroux3-cie}).
Knowing who wins the game allows for constructions such as the one used in
Lemma \ref{lemma:findWSparallelizable} to conclude that finding a winning
strategy is parallelizable (i.e.\ $\widehat{\FindWSS} \equivW \FindWSS$ and
$\widehat{\FindWSP} \equivW \FindWSP$). We will see in Corollary
\ref{corr:detsigmanparallel} below that this is not just an obstacle for the
proof strategy, but that the result differs for $\DetS$.

\subsubsection*{If then else}

As we have seen, many theorems equivalent to ${\rm ATR}_0$ are described as
{\em dichotomy}-type theorems: Exactly one of $A$ or $B$ holds. Thus, it is
natural to consider the following if-then-else problem for a given dichotomy
$A$ xor $B$: Provide two descriptions $(\alpha,\beta)$ trying to verify $A$
and $B$ simultaneously. If $A$ is true, then $\alpha$ is a correct proof
validating $A$; or else $\beta$ is a correct proof of $B$, where we do not
need to know which one is correct. We formalize this idea as follows.

A space of truth values is just a represented space $\mathbb{B}$ with underlying set $\{\top,\bot\}$.

\begin{definition}
Let $\mathbb{B}$ be a space of truth values. For $f : \subseteq \mathbf{X} \mto \mathbf{Y}$ and $g : \subseteq \mathbf{A} \mto \mathbf{B}$, we define
\[ \Wif{\mathbb{B}}{f}{g} : \subseteq \mathbb{B} \times \mathbf{X} \times \mathbf{A} \mto \mathbf{Y} \times \mathbf{B}\]
via $(b,x_0,x_1) \in \dom(\Wif{\mathbb{B}}{f}{g})$ iff $b = \top$ and $x_0 \in \dom(f)$ or $b = \bot$ and $x_1 \in \dom(g)$, and $(y_0,y_1) \in \Wif{\mathbb{B}}{f}{g}(b,x_0,x_1)$ iff $b = \top$ and $y_0 \in f(x_0)$ or $b = \bot$ and $y_1 \in g(x_1)$.
\end{definition}

Note that the degree of $\Wif{\mathbb{B}}{f}{g}$ depends on the precise
choice of spaces for domain and codomains involved, beyond what matters for
where $f$ and $g$ are actually defined and are taking their range. In
particular, $\Wif{\mathbb{B}}{f}{g}$ is not an operation on Weihrauch
degrees\footnote{Let $\mathbf{X}$ be the represented space of the
non-computable elements of \Bai, and $f: \subseteq \Bai \to \Bai$ the
restriction of $\id_\Bai$ to the non-computable elements ($\id_\mathbf{X}$
and $f$ are the same function, but defined on different spaces); then
$\id_\mathbf{X} \equivW f$, yet $\Wif{\mathbb{S}}{f}{\id_\Bai} \nleqW
\Wif{\mathbb{S}}{\id_\mathbf{X}}{\id_\Bai}$ because the former has computable
inputs while the latter does not.}.

\subsubsection*{The upper bound}
Let $\mathbb{S}_{\SI11}$ be the space of truth values where $p$ is a name for
$\top$ iff $p$ codes an ill-founded tree, and a name for $\bot$ iff it codes
a well-founded tree.

In the proofs of Propositions \ref{prop:ptt1} and \ref{prop:findwspi}, we
constructed closed sets containing information over the perfect subtrees or
the winning strategies of Player 2 respectively. In particular, by testing
whether these are empty or not, we can decide in which case we are, and
obtain the answer in $\mathbb{S}_{\SI11}$. Thus, by combining Proposition
\ref{prop:ptt1} and Theorem \ref{theo:list}, respectively Proposition
\ref{prop:findwspi} and Theorem \ref{theo:findWSSigma}, we obtain the
following:

\begin{corollary}\label{cor:ptt-ifthenelse}
$\PTT2 \leqW \Wif{\mathbb{S}_{\SI11}}{\C_\Bai}{\UC_\Bai}$.
\end{corollary}

\begin{corollary}\label{cor:det-ifthenelse}
$\DetS \leqW \Wif{\mathbb{S}_{\SI11}}{\C_\Bai}{\UC_\Bai}$.
\end{corollary}

As $\UC_\Bai \leqW \C_\Bai$, it follows that
$\Wif{\mathbb{S}_{\SI11}}{\C_\Bai}{\UC_\Bai} \leqW \C_\Bai \star \chiP$. In
particular, the difference between
$\Wif{\mathbb{S}_{\SI11}}{\C_\Bai}{\UC_\Bai}$ and $\C_\Bai$ disappears if we
move from Weihrauch reducibility to computable reducibility. It follows
immediately that Gandy's basis theorem applies to $\DetS$: Every $\SI01$-game
has a Nash equilibrium that is hyperlow relative to the game.

\subsubsection*{Idempotency}
We can show a kind of absorption result for the if-then-else construction.
Recall that $\mathsf{NHA}$ asks for an output that is not hyperarithmetic
relative to the input.
\begin{proposition}
\label{prop:nahabsorptions} Let $g$ have a hyperarithmetical point $\rho$ in
its codomain. If we have $f \times \mathsf{NHA} \leqW
\Wif{\mathbb{B}}{g}{\UC_\Bai}$, then $f \leqW g$.
\begin{proof}
Any $x \in \dom(f)$ is provided in the form of some name $p_x$, which is a
valid input to $\mathsf{NHA}$. If some $(x,p_x) \in \dom(f \times
\mathsf{NHA})$ were mapped to some $(\bot,a,A)$ via the reduction, then $A =
\{q\}$ where $q$ is hyperarithmetical in $p_x$. Then $(\rho,q)$ is a valid
output of $\Wif{\mathbb{B}}{g}{\UC_\Bai}$, but we cannot compute a solution
to $\mathsf{NHA}(p_x)$ from $(\rho,q)$.

Thus, every $(x,p_x)$ gets mapped to $(\top,a_x,A)$ such that from $b \in
g(a_x)$ we can compute $y \in f(x)$ (since $(b,z)$ for any $z$, say
$(b,\emptyset)$, is a solution to the instance $(\top,a_x,A)$). This provides
the claimed reduction $f \leqW g$.
\end{proof}
\end{proposition}

By Corollaries \ref{cor:two-sided-det-ptt}, \ref{cor:det-ifthenelse} and
\ref{cor:ptt-ifthenelse}, and Proposition \ref{prop:nahabsorptions} we get
the following:

\begin{corollary}
$\wPTT2 \times \mathsf{NHA} \nleqW
\Wif{\mathbb{S}_{\SI11}}{\C_\Bai}{\UC_\Bai}$.
\end{corollary}

\begin{corollary}
$\DetS \times \mathsf{NHA} \nleqW
\Wif{\mathbb{S}_{\SI11}}{\C_\Bai}{\UC_\Bai}$.
\end{corollary}

Using the corollaries above in conjunction with Corollary
\ref{corr:nhabasic}, we obtain:

\begin{corollary}
$\wPTT2 \times \C_\Bai \nleqW \PTT2$ and hence $\wPTT2 \times \wPTT2 \nleqW
\PTT2$.
\end{corollary}

\begin{corollary}\label{corr:detsigmanparallel}
$\DetS \times \C_\Bai \nleqW \DetS$ and hence $\DetS \times \DetS \nleqW
\DetS$.
\end{corollary}

\subsubsection*{Products with $\UC_\Bai$}
While we just saw that $\DetS$, $\PTT2$ and
$\Wif{\mathbb{S}_{\SI11}}{\C_\Bai}{\UC_\Bai}$ are not closed under products
with $\C_\Bai$, the situation for products with $\UC_\Bai$ is different:

\begin{proposition}
\label{prop:ucbaireproduct1} $\UC_\Bai \times
\Wif{\mathbb{B}}{\C_\Bai}{\UC_\Bai} \equivW
\Wif{\mathbb{B}}{\C_\Bai}{\UC_\Bai}$ for any space of truth values
$\mathbb{B}$.
\begin{proof}
Let $\{a\}$, $b \in \mathbb{B}$, $A$, $B$ be the input to $\UC_\Bai \times
\Wif{\mathbb{B}}{\C_\Bai}{\UC_\Bai}$. We can use
$\Wif{\mathbb{B}}{\C_\Bai}{\UC_\Bai}$ on $b$, $\{a\} \times A$ and $\{a\}
\times B$, as $\{a\} \times A$ is non-empty iff $A$ is, and $\{a\} \times B$
is a singleton iff $B$ is. We will receive as output $(\langle p, x\rangle,
\langle q, y\rangle)$ such that $\langle x, y\rangle$ is a valid output to
$\Wif{\mathbb{B}}{\C_\Bai}{\UC_\Bai}(b,A,B)$, and at least one of $p$ and $q$
is $a$. Let us write $p_{\leq n}$ for the prefix of $p$ of length $n+1$. We
have that, if $p_{\leq n} = q_{\leq n}$, then $p_{\leq n} = a_{\leq n}$, and
if $p_{\leq n} \neq q_{\leq n}$, then either $p \notin \{a\}$ or $q \notin
\{a\}$, hence we can compute $a$ from $p$, $q$ and $\{a\}$.
\end{proof}
\end{proposition}

\begin{proposition}
\label{prop:pttproductuc} $\UC_\Bai \times \PTT2 \equivW \PTT2$.
\begin{proof}
Let $(\{a\}, T)$ be the input to $\UC_\Bai \times \PTT2$. From this input we
can build a tree $T_0$ such that $[T_0] = \{a\} \times (\{0^\omega\} \cup
1[T])$ (notice that $|[T_0]| =|[T]|+1$). $\PTT2 (T_0)$ yields a tree $T'$ and
a sequence $n\langle (q_0, t_0p_0), (q_1, t_1p_1), \ldots \rangle$.

We first explain how to compute the sequence part of $\PTT2 (T)$. If $n = 1$,
or $n = 0$ and more than one $t_i$ is $0$, or $n>1$ and more than one $t_i$
for $i < n - 1$ is $0$, then the sequence is not listing $[T_0]$ (because
$[T_0] \neq \emptyset$ and $(a,0^\omega)$ is the only member of $[T_0]$ whose
second component starts with $0$), which implies that $[T_0]$, and hence
$[T]$, was uncountable. In this case, we can just output some arbitrary
sequence. Otherwise let $p'_i$ be the sequence consisting of the odd digits
of $p_i$. If $n = 0$, we output $0 \langle p'_{i_0}, p'_{i_1}, \ldots
\rangle$ where the $i_k$ are the (all but one) indices such that $t_i \neq 0$
(in this way, if $\langle (q_0, t_0p_0), (q_1, t_1p_1), \ldots \rangle$ lists
injectively $[T_0]$, our output lists injectively $[T]$). To achieve the same
result when $n > 1$ we output $(n-1) \langle p'_{i_0}, p'_{i_1},
\ldots\rangle$ where we are omitting the (at most one) $i<n-1$ such that $t_i
= 0$.

To compute the tree part of $\PTT2 (T)$, starting from $T'$ we obtain a tree
$T''$ as follows: On the first three levels (corresponding to the first two
digits of $a$ and the control bit), go down some arbitrary edge in $T'$. Then
alternate adding all children of the present vertices into $T''$, and passing
down some arbitrary edge. If $T'$ is perfect, then so is $T''$, and moreover,
$T'' \subseteq T$ in that case.

We need also to compute $a$. To produce a possible candidate, we attempt to
compute the left-most branch $q$ of $T'$. If we ever reach a leaf (which
never happens if $T'$ is perfect), then we continue $q$ by constant $0$. In
any case, let $q'$ be the even digits of $q$: if $T'$ is a perfect subtree of
$T_0$ then $a=q'$. On the other hand, if $(q_0, t_0p_0), (q_1, t_1p_1),
\ldots \rangle$ lists $[T_0]$ then $a=q_0$. Thus $a = q_0$ or $a = q'$. As in
the proof of Proposition \ref{prop:ucbaireproduct1} it follows that we can
compute $a$ from $q_0$, $q'$ and $\{a\}$.
\end{proof}
\end{proposition}

\begin{proposition}
\label{prop:detproductuc} $\UC_\Bai \times \DetS \equivW \DetS$.
\begin{proof}
By Theorem \ref{theo:findWSSigma}, we have $\UC_\Bai \leqW \FindWSD$, i.e.~we
can compute a $\DE01$-game $G'_1$ from $\{a\}$ such that Player $1$ wins
$G'_1$, and from a winning strategy of Player $1$ in $G'_1$ we can compute
$a$. Let $G'_2$ be the game with the roles of Player $1$ and Player $2$
exchanged, which is still \DE01. Now we construct a \SI01 game $G''$ from a
$\SI01$-game $G$, and from $G'_1$ and $G'_2$.

The players start playing $G$ and $G'_2$ in parallel. If Player $2$ wins both
of these, he wins in $G''$. Else, if he loses one of them (which would happen
at some finite time), the players proceed to play $G'_1$, and whoever wins
$G'_1$ wins $G''$. W.l.o.g.~we assume that Player $2$ can choose to lose $G$
right at the start of $G''$.

Since by assumption Player $2$ has a winning strategy in $G'_2$, and Player
$1$  has a winning strategy in $G'_1$, the winning strategies of Player $2$
are exactly those that consists of playing winning strategies in $G$ and
$G'_2$ simultaneously. On the other hand, Player $1$ can win the game for
sure only by first playing a winning strategy in $G$ (and arbitrarily in
$G_2'$), followed by a winning strategy in $G'_1$.

From a Nash equilibrium of the whole game we thus obtain a Nash equilibrium
in $G$ by considering how the players play in $G$. Furthermore, we consider
how Player $1$ plays in the copy of $G'_1$ played when Player $2$ loses in
$G$ right at the start of $G''$, and how Player $2$ plays in $G'_2$, and
compute two candidates $q_0$, $q_1$ for $a$ from that. As in the proof of
Proposition \ref{prop:ucbaireproduct1}, we can then compute $a$ from $\{a\}$,
$q_0$ and $q_1$.
\end{proof}
\end{proposition}

Here the difference between $\wPTT2$ and $\PTT2$ is revealed, as the former
is more sensitive to products. We recall that a Weihrauch degree is called
\emph{fractal}, if it has a representative $f : \subseteq \Bai \mto \Bai$
such that for any $w \in \N^{<\N}$ such that $w\Bai \cap \dom(f) \neq
\emptyset$ it holds that $f|_{w\Bai} \equivW f$. Most of the degrees
considered in this articles are fractals, including $\wPTT2$.

\begin{proposition}
If $f$ is a fractal and $\lpo \times f \leqW \wPTT2$, then $f \leqW \C_\Bai$.
\begin{proof}
W.l.o.g.~assume that $f : \subseteq \Bai \mto \Bai$ witnesses its own
fractality.

Fix a reduction of $\lpo \times f$ to $\wPTT2$ and let $K_1$ be the
computable function that transforms the output of $\wPTT2$ and the original
input of $\lpo \times f$ into the answer to the $\lpo$-instance. We
distinguish the following cases:
\begin{enumerate}
\item There exists $0^n$, $w \in \N^{<\N}$, a finite tree $T$, and a finite
    prefix of a list $\langle 0q_0,0q_1,0q_2,\ldots\rangle$ such that $K_1$
    provides its answer upon reading those (as input for $\lpo$, input for
    $f$, first and second component of the output of $\wPTT2$, in that
    order).

Then by fixing the input to $\lpo$ to something consistent with $0^n$ and
incompatible with the answer provided, we can make sure that the reduction
needs to avoid the prefix to be valid for any input to $f$ extending $w$.
But this can only be achieved by making the input to $\wPTT2$ having
uncountable body and not having $T$ as prefix of any perfect subtree. This
means in particular that we are dealing with an input to $\PTT1$. As $f$ is
a fractal, restricting to those of its inputs extending $w$ does not
decrease its Weihrauch degree, and we conclude $f \leqW \C_\Bai$.

\item For no $0^n$, $w \in \N^{<\N}$, finite tree $T$, and finite prefix of
    a list $\langle 0q_0,0q_1,0q_2,\ldots\rangle$, $K_1$ provides its
    answer upon reading those.

If we fix the $\lpo$-input to be $0^\omega$, we see that to ensure that
$K_1$ behaves correctly, the list-component of the output of $\wPTT2$ must
actually list some elements. This can only be guaranteed if the input to
$\wPTT2$ is a tree with countable non-empty body, i.e.~is already in the
domain of $\List$. We thus conclude $f \leqW \List \equivW \UC_\Bai$ (by
Theorem \ref{theo:list}) and, a fortiori, $f \leqW \C_\Bai$.\qedhere
\end{enumerate}
\end{proof}
\end{proposition}

\begin{corollary}
\label{corr:lpowptt2} $\lpo \times \wPTT2 \nleqW \wPTT2$.
\end{corollary}

\begin{corollary}
$\wPTT2 \leW \PTT2$.
\begin{proof}
By contrasting Corollary \ref{corr:lpowptt2} and Proposition \ref{prop:pttproductuc}.
\end{proof}
\end{corollary}

We shall see that $\wPTT2$ is still closed under some non-trivial products.
For that, let $\mathsf{NON} : \Can \mto \Can$ be defined via $q \in
\mathsf{NON}(p)$ iff $q \nleq_{\mathrm{T}} p$; i.e.~$\mathsf{NON}$ is the
function corresponding to the theorem asserting the existence of sets
non-computable in any given set.

\begin{proposition}
\label{prop:nonwptt2} $\mathsf{NON} \times \wPTT2 \leqW \wPTT2$.
\begin{proof}
Fix a Turing functional $\Phi$ such that for every $p \in \Can$, $\Phi^p$ is
an injective enumeration of $p'$, the Turing jump of $p$. Let $\hat{p} \in
\Bai$ be such that for every $n$ we have that $\hat{p}(n)=0$ implies $n
\notin p'$ and $\hat{p}(n)>0$ implies $\Phi^p(p(n)-1) = n$. Then $\hat{p}$ is
Turing equivalent to $p'$ and hence $\hat{p} \nleq_{\mathrm{T}} p$.

Notice that the function from \Can\ to $\mathcal{A}(\Bai)$ which sends $p$ to
$\{\hat{p}\}$ is computable. Therefore, from $(p,A) \in \Can \times
\mathcal{A}(\Bai)$ we can compute $\{\hat{p}\} \times (\{0^\omega\} \cup 1A)
\in \mathcal{A}(\Bai)$. From any solution to $\wPTT2(\{\hat{p}\} \times
(\{0^\omega\} \cup 1A))$ we can compute a solution to $\wPTT2(A)$ with the
argument of the first part of the proof of Proposition
\ref{prop:pttproductuc}. Moreover, any solution to $\wPTT2(\{\hat{p}\} \times
(\{0^\omega\} \cup 1A))$ is $\geq_{\mathrm{T}} \hat{p}$, and hence solves
$\mathsf{NON}(p)$.
\end{proof}
\end{proposition}

In \cite{pauly-patey}, products with $\lpo$ and $\mathsf{NON}$ are used to
separate Weihrauch degrees in a similar fashion.

\section{$\TC_\Bai$ -- a candidate for $\mathrm{ATR}_0$?}
\label{sec:tcbaire} Our separation proofs of principles like $\DetS$ and
$\PTT2$ from $\C_\Bai$ relied on being able to transform an arbitrary closed
subset into an input for the former, with specified behaviour occurring only
for non-empty closed sets. We can capture this using the notion of
\emph{total continuation} of closed choice on $\Bai$:

\begin{definition}
Let $\TC_\Bai : \mathcal{A}(\Bai) \mto \Bai$ be defined via $p \in
\TC_\Bai(A)$ iff $A \neq \emptyset \Rightarrow p \in A$.
\end{definition}

In the same vein, we can define the total continuation of other choice
principles. The computable compactness of $\Can$ yields $\TC_\Can \equivW
\C_\Can$. The principle $\TC_\N$ was studied in \cite{paulyneumann}.

\begin{proposition}\label{prop:TC}
\begin{enumerate}
\item $\C_\Bai \leW \TC_\Bai$;
\item $\TC_\Bai \leW \lpo \times \TC_\Bai$.
\item If $\mathsf{NON} \times f \leqW \TC_\Bai$, then $f \leqW \C_\Bai$;
\item $\TC_\Bai \leW \wPTT2$;
\item $\TC_\Bai \leW \DetS$;
\item $\Wif{\mathbb{S}_{\SI11}}{\C_\Bai}{\UC_\Bai} \leW \TC_\Bai \times
    \C_\Bai$.
\end{enumerate}
\begin{proof}
\begin{enumerate}
\item The reduction is trivial. Separation follows from $\lpo \star \C_\Bai
    \equivW \C_\Bai$ and $\chiP \leqW \lpo \star \TC_\Bai$ (the latter is
    straightforward because $\lpo$ can check whether the output of
    $\TC_\Bai(A)$ belongs to $A$).
\item Again, the reduction is trivial. For the separation, assume that
    $\lpo \times \TC_\Bai \leqW \TC_\Bai$ via computable $H$, $K_1$, $K_2$.
    Recall that $\lpo(r)=1$ iff $r=0^\omega$. Consider the input $0^\omega$
    for $\lpo$ and $\Bai \in \mathcal{A}(\Bai)$ (coded as some name $t$)
    for $\TC_\Bai$ on the left. There has to be some $p \in \Bai$ such that
    $K_1(0^\omega,t,p) = 1$. By continuity, we find that $K_1(0^kq,t_{\leq
    k}t',p) = 1$ for sufficiently large $k$ and arbitrary $q$, $t'$.

    For any $A \in \mathcal{A}(\Bai)$ we can compute some name of the form
$t_{\leq k}t'$. Now consider what happens if the inputs on the left are
$0^k1^\omega$ and some $t_{\leq k}t'$: If $H(0^k1^\omega,t_{\leq k}t')$
ever returns a name for the empty set, then $p$ is a valid solution to
$\TC_\Bai$ on the right. But then $K_1$ will answer incorrectly $1$. Thus,
$H(0^k1^\omega,t_{\leq k}t')$ never returns a name for the empty set. But
then we obtain a reduction $\TC_\Bai \leqW \C_\Bai$, contradicting $(1)$.
\item As $\TC_\Bai(\emptyset)$ has computable solutions, the reduction
    $\mathsf{NON} \times f \leqW \TC_\Bai$ already has to be a reduction to
    $\C_\Bai$.
\item The reduction given in Proposition \ref{prop:ptt1} works for this, by
    using the following observation: given $A \in \mathcal{A}(\Bai)$, $T
    \in \Tr$ such that $[T] =  A \times \Bai$ and $(T',\langle
    b_0p_0,b_1p_1,\ldots \rangle) \in \PTT2(T)$, if we realize that $T'$ is
    not pruned (which can happen only if $A = \emptyset$) we can continue
    our output with $0^\omega$.

    Strictness follows by $(3)$, Proposition \ref{prop:nonwptt2} and
    Corollary \ref{cor:two-sided-det-ptt}.
\item The reduction given in Proposition \ref{prop:findwspi} works for
    this, by using the following observation: if $A = \emptyset$ then
    Player $1$ has a winning strategy in the \SI01 game we constructed (in
    fact, any strategy for $1$ is winning), however following the strategy
    for $2$ provided by $\DetS$ we find an element of $\TC_\Bai(A)$.

    Strictness follows by $(2)$, Proposition \ref{prop:detproductuc} and
    Corollary \ref{cor:two-sided-det-ptt}.
\item The arguments used to establish Lemma \ref{lemma:globalmcb} or
    \ref{lemma:findWSSigma} show that the total continuation $\TUC_\Bai$ of
    $\UC_\Bai$ (i.e.~the total multivalued function defined on
    $\mathcal{A}(\Bai)$ which extends $\UC_\Bai$ and is defined as $\Bai$
    on non-singletons) is reducible to $\C_\Bai$. For example, given an
    arbitrary closed $A \subseteq \Bai$ we can compute the nonempty
    $\Sigma^1_1$ set of the mCB-certificates of $A$ and, choosing an
    element in it, compute the list of the elements of $A$ whenever $A$ is
    a countable, and in particular a singleton.

    Thus, we can consider $\TC_\Bai \times \TUC_\Bai$ in place of $\TC_\Bai
    \times \C_\Bai$. Given some input $b, A, B$ to
    $\Wif{\mathbb{S}_{\SI11}}{\C_\Bai}{\UC_\Bai}$ we ignore $b$, we feed
    $A$ to $\TC_\Bai$, and $B$ to $\TUC_\Bai$. Any resulting output pair is
    a valid output to $\Wif{\mathbb{S}_{\SI11}}{\C_\Bai}{\UC_\Bai}$.

    To see that $\TC_\Bai \times \C_\Bai \not\equivW
    \Wif{\mathbb{S}_{\SI11}}{\C_\Bai}{\UC_\Bai}$ first notice that
    $\TC_\Bai \times \C_\Bai \times \C_\Bai \equivW \TC_\Bai \times
    \C_\Bai$. On the other hand, we have
    \[
\Wif{\mathbb{S}_{\SI11}}{\C_\Bai}{\UC_\Bai} \times \C_\Bai \nleqW \Wif{\mathbb{S}_{\SI11}}{\C_\Bai}{\UC_\Bai}:
    \]
    otherwise, since by Corollaries \ref{cor:ptt-ifthenelse} and
    \ref{corr:nhabasic} we have $\PTT2 \leqW
    \Wif{\mathbb{S}_{\SI11}}{\C_\Bai}{\UC_\Bai}$ and $\mathsf{NHA} \leqW
    \C_\Bai$, we would have $\PTT2 \times \mathsf{NHA} \leqW
    \Wif{\mathbb{S}_{\SI11}}{\C_\Bai}{\UC_\Bai}$ and Proposition
    \ref{prop:nahabsorptions} would imply $\PTT2 \leqW \C_\Bai$, against
    Corollary \ref{cor:two-sided-det-ptt}.\qedhere
\end{enumerate}
\end{proof}
\end{proposition}

\begin{corollary}
$\PTT2^* \equivW \DetS^* \equivW \TC_\Bai^*$.
\end{corollary}
\begin{proof}
$\TC_\Bai^* \leqW \PTT2^*$ is immediate from Proposition \ref{prop:TC}(4). On
the other hand we have
\[
\PTT2^* \leqW \Wif{\mathbb{S}_{\SI11}}{\C_\Bai}{\UC_\Bai}^* \leqW (\TC_\Bai \times \C_\Bai)^* \leqW \TC_\Bai^*,
\]
using Corollary \ref{cor:ptt-ifthenelse} and Proposition \ref{prop:TC}(6).

The argument for $\DetS^*$ is similar.
\end{proof}

It is reasonable to expect a Weihrauch degree corresponding to an axiom
system from reverse mathematics to be closed under finite parallelization.
For candidates for $\mathrm{WKL}_0$ or $\mathrm{ACA}_0$ this happens
inherently. Here, we might need to demand it explicitly, and thus consider
the degree $\TC_\Bai^*$ rather than any directly defined one to be one of the
most promising candidates.

A potentially convenient way to think about the separation between $\C_\Bai$
and $\TC_\Bai$ is in terms of translations between truth values. $\TC_\Bai$
allows us to treat a single $\PI11$-set as an open set, whereas $\C_\Bai$
cannot even bridge the gap from $\SI11$ to Borel.

\begin{proposition}
$\left (\id : \mathbb{S}_{\PI11} \to \mathbb{S} \right ) \leqW \TC_\Bai$, but
$\idStB \nleqW \C_\Bai$.
\begin{proof}
For the reduction, we observe that $A = \emptyset$ iff $p \notin A$ for some
$p \in \TC_\Bai(A)$.

For the non-reduction, we recall that $\id : \mathbb{S}_\mathcal{B} \to
\mathbf{2} \leqW \UC_\Bai$ was shown in \cite[Lemma 79]{pauly-ordinals}, and that
$\UC_\Bai \star \C_\Bai \equivW \C_\Bai$ as shown in \cite[Theorem 7.3]{paulybrattka}.
Thus, assuming the reduction would hold, we would even have that $\left (\id
: \mathbb{S}_{\PI11} \to \mathbf{2} \right ) \leqW \C_\Bai$, which
contradicts \cite[Theorem 7.7]{paulybrattka} because the unique realizer of $\id :
\mathbb{S}_{\PI11} \to \mathbf{2}$ is not effectively Borel measurable.
\end{proof}
\end{proposition}

Next, we shall see that the additional computational power of $\TC_\Bai$
(even of its parallelization $\widehat{\TC_\Bai}$) over $\UC_\Bai$ concerns
only multivalued problems.

\begin{theorem}\label{thm:parallel-total-continuation}
The following are equivalent for single-valued $f: \subseteq \mathbf{X} \to
\Bai$ where $\mathbf{X}$ is a represented space:
\begin{enumerate}
\item $f \leqW \UC_\Bai$;
\item $f \leqW \widehat{\TC_\Bai}$.
\end{enumerate}
\begin{proof}
That $1$ implies $2$ is trivial. For the other direction, we first argue that it suffices to consider single-valued $f : \subseteq \Bai \to \{0,1\}$. Then we show that for
single-valued $f : \subseteq \Bai \to \{0,1\}$, $f \leqsW \widehat{\TC_\Bai}$
implies $f \leqW \DeCA11$ and invoke Theorem \ref{theo:bigucbaire}.

Let $\delta_\mathbf{X}$ be the representation of $\mathbf{X}$. For $f :
\mathbf{X} \to \Bai$, consider the map $F :\subseteq \Bai \to \{0,1\}$ where
$F(nmp) = 1$ if $f(\delta_\mathbf{X}(p)) (n) = m$ and $F(nmp) = 0$ otherwise,
provided $p \in \dom(f\delta_\mathbf{X})$. Now it holds that $F \leqW f \leqW
\widehat{F}$ (the latter reduction holds because $f$ is single-valued). As
$\UC_\Bai$ is parallelizable, $F \leqW \UC_\Bai$ is equivalent to
$\widehat{F} \leqW \UC_\Bai$ and hence $f \leqW \UC_\Bai$.

For the second claim, we can start from a strong Weihrauch reduction because
$\widehat{\TC_\Bai}$ is a cylinder. Assume that $f : \subseteq \Bai \to \{0,1\}$ and $f \leqsW
\widehat{\TC_\Bai}$ via computable $K$, $H$. The outer reduction witness $K$
essentially consists of two open sets $U^0, U^1 \in \mathcal{O}((\Bai)^\N)$,
while the inner reduction witness $H$ gives us for each $p \in \Bai$ a
sequence $(A_n(p))_{n \in \N}$ of closed sets. For $S \subseteq \N$ and $U
\in \mathcal{O}((\Bai)^\N)$, let $\pi_S(U)$ denote the projection of $U$ to
the components in $S$. Now we find that:
\[
f(p) = b \Leftrightarrow \forall S \subseteq \N \ \prod_{n \in S} A_n(p)
\subseteq \pi_S(U^b).
\]
(Notice that $\prod_{n \in \N} A_n(p) \subseteq U^b$ does not imply $f(p) =
b$ in general because some of the $A_n(p)$ could be empty.) This is a
$\PI11$-condition. Since exactly one of $f(p) = 0$ and $f(p) = 1$ holds, we
thus have a valid instance for $\DeCA11$.
\end{proof}
\end{theorem}

In particular, $\widehat{\TC_\Bai}$ does not reach the level of
$\Pi^1_1\mbox{-}\mathrm{CA}_0$.

\begin{corollary}\label{corr:chisigma11}
$\chiP \nleqW \widehat{\TC_\Bai}$.
\end{corollary}

\section{Open questions and discussion}
\label{sec:discussion}
The results reported in Section \ref{sec:twosided} immediately lead to three interlinked questions, which unfortunately we have been unable to resolve so far:

\begin{question}
Does $\DetS \equivW \Wif{\mathbb{S}_{\SI11}}{\C_\Bai}{\UC_\Bai}$?
\end{question}

\begin{question}
Does $\PTT2 \equivW \Wif{\mathbb{S}_{\SI11}}{\C_\Bai}{\UC_\Bai}$?
\end{question}

\begin{question}
How do $\PTT2$ and $\DetS$ relate?
\end{question}

We would expect that other theorems equivalent to $\mathrm{ATR}_0$ (e.g.\
open Ramsey) exhibit similar behaviour, i.e.~a non-constructive disjunction
between cases equivalent to $\C_\Bai$ and $\UC_\Bai$ respectively. Proving
any reductions between the two-sided versions of these theorems could be very
illuminating. Until then, we might have to settle for classifications in the
Weihrauch lattice up to $^*$, and strive to understand better the degree
$\TC_\Bai^*$.

Brattka has also raised the question whether the strong two-sided versions,
which return an answer on the applicable case together with a witness, are
worthwhile studying. It seems conceivable that finding reductions here would
be easier. Up to $^*$, these problems would have the degree $\TC_\Bai^*
\times \chiP^*$. Would this be an acceptable candidate for an
$\mathrm{ATR}_0$-equivalent, or is this degree too close to
$\PI11\mbox{-}\mathrm{CA}_0$?

Given that $\TC_\Bai^*$ is not closed under composition, one could make the
case that $\TC_\Bai^\diamond$ (its closure under generalized register
machines, cf.~\cite{paulyneumann}) is the better candidate. Note that
$\TC_\Bai^\diamond \equivW \left (\TC_\Bai \times \chiP \right )^\diamond$,
so the distinction between the weak and strong two-sided versions of the
theorems would disappear here. How well justified this step would be in
particular depends on whether there exists a natural theorem equivalent to
$\mathrm{ATR}_0$ in reverse mathematics where $\mathrm{ATR}_0$ is actually
used in a sequential way, i.e.~a theorem naturally associated with a
Weihrauch degree not reducible to $\TC_\Bai^*$.

\section*{Acknowledgements}

In the earlier stages of this research Marcone collaborated with Andrea Cettolo, and some of the proofs were obtained jointly with him. Pauly began working on this project while being a visiting fellow at the Isaac Newton Institute for Mathematical Sciences in the programme `Mathematical, Foundational and Computational Aspects of the Higher Infinite'. He thanks Vasco Brattka, Jun Le Goh, Luca San Mauro and Richard Shore for inspiring conversations. The research project benefitted from discussion at the Dagstuhl seminars  15392 `Measuring the Complexity of Computational Content: Weihrauch Reducibility and Reverse Analysis' and 18361 `Measuring the Complexity of Computational Content: From Combinatorial Problems to Analysis'.

Kihara's research was partially supported by JSPS KAKENHI Grant 17H06738, 15H03634, and the JSPS Core-to-Core Program (A. Advanced Research Networks).
Marcone's research was partially supported by PRIN 2012 Grant \lq\lq Logica, Modelli e Insiemi\rq\rq\ and by the departmental PRID funding \lq\lq HiWei --- The higher levels of the Weihrauch hierarchy\rq\rq.

\noindent\begin{minipage}{0.1\textwidth}\includegraphics[width=\textwidth]{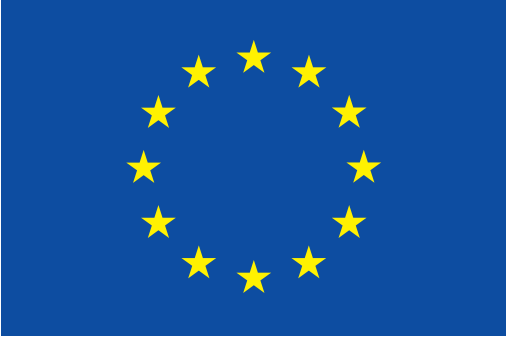}\end{minipage}
\begin{minipage}{0.9\textwidth} Pauly has received funding from the European Union’s Horizon 2020 research and innovation programme under the Marie Sklodowska-Curie grant agreement No 731143, \emph{Computing with Infinite Data}.\end{minipage}

\bibliographystyle{eptcs}
\end{document}